\documentclass[pra,aps,showpacs,onecolumn,twoside,superscriptaddress]{revtex4}

\usepackage{amsmath,amsfonts,amssymb,caption,color,epsfig,graphics,graphicx,hyperref,latexsym,mathrsfs,revsymb,theorem,url,verbatim,epstopdf,booktabs,multirow,tabularx,threeparttable,subfigure}
\usepackage{algorithm}
\usepackage{algorithmic}

\hypersetup{colorlinks,linkcolor={blue},citecolor={red},urlcolor={blue}}

\newtheorem{definition}{Definition}
\newtheorem{proposition}[definition]{Proposition}
\newtheorem{lemma}[definition]{Lemma}

\newtheorem{theorem}[definition]{Theorem}
\newtheorem{corollary}[definition]{Corollary}
\newtheorem{conjecture}[definition]{Conjecture}

\newtheorem{remark}[definition]{Remark}
\newtheorem{example}[definition]{Example}
\newtheorem{question}[definition]{Question}

\def\bcj{\begin{conjecture}}
\def\ecj{\end{conjecture}}
\def\bcr{\begin{corollary}}
\def\ecr{\end{corollary}}
\def\bd{\begin{definition}}
\def\ed{\end{definition}}
\def\bea{\begin{eqnarray}}
\def\eea{\end{eqnarray}}
\def\bem{\begin{enumerate}}
\def\eem{\end{enumerate}}
\def\bex{\begin{example}}
\def\eex{\end{example}}
\def\bim{\begin{itemize}}
\def\eim{\end{itemize}}
\def\bl{\begin{lemma}}
\def\el{\end{lemma}}
\def\bma{\begin{bmatrix}}
\def\ema{\end{bmatrix}}
\def\bpf{\begin{proof}}
\def\epf{\end{proof}}
\def\bpp{\begin{proposition}}
\def\epp{\end{proposition}}
\def\bqu{\begin{question}}
\def\equ{\end{question}}
\def\br{\begin{remark}}
\def\er{\end{remark}}
\def\bt{\begin{theorem}}
\def\et{\end{theorem}}

\def\squareforqed{\hbox{\rlap{$\sqcap$}$\sqcup$}}
\def\qed{\ifmmode\squareforqed\else{\unskip\nobreak\hfil
\penalty50\hskip1em\null\nobreak\hfil\squareforqed
\parfillskip=0pt\finalhyphendemerits=0\endgraf}\fi}
\def\endenv{\ifmmode\;\else{\unskip\nobreak\hfil
\penalty50\hskip1em\null\nobreak\hfil\;
\parfillskip=0pt\finalhyphendemerits=0\endgraf}\fi}
\newenvironment{proof}{\noindent \textbf{{Proof.~} }}{\qed}
\def\Dbar{\leavevmode\lower.6ex\hbox to 0pt
{\hskip-.23ex\accent"16\hss}D}
\makeatletter
\def\url@leostyle{%
  \@ifundefined{selectfont}{\def\UrlFont{\sf}}{\def\UrlFont{\small\ttfamily}}}
\makeatother
\def\bcj{\begin{conjecture}}
\def\ecj{\end{conjecture}}
\def\bcr{\begin{corollary}}
\def\ecr{\end{corollary}}
\def\bd{\begin{definition}}
\def\ed{\end{definition}}
\def\bea{\begin{eqnarray}}
\def\eea{\end{eqnarray}}
\def\bem{\begin{enumerate}}
\def\eem{\end{enumerate}}
\def\bex{\begin{example}}
\def\eex{\end{example}}
\def\bim{\begin{itemize}}
\def\eim{\end{itemize}}
\def\bl{\begin{lemma}}
\def\el{\end{lemma}}
\def\bpf{\begin{proof}}
\def\epf{\end{proof}}
\def\bpp{\begin{proposition}}
\def\epp{\end{proposition}}
\def\bqu{\begin{question}}
\def\equ{\end{question}}
\def\br{\begin{remark}}
\def\er{\end{remark}}
\def\bt{\begin{theorem}}
\def\et{\end{theorem}}

\def\btb{\begin{tabular}}
\def\etb{\end{tabular}}

\newcommand{\nc}{\newcommand}

\def\a{\alpha}

\def\e{\epsilon}

\def\i{\iota}

\def\l{\lambda}

\def\s{\sigma}

 \nc{\bbA}{\mathbb{A}} \nc{\bbB}{\mathbb{B}} \nc{\bbC}{\mathbb{C}}
 \nc{\bbD}{\mathbb{D}} \nc{\bbE}{\mathbb{E}} \nc{\bbF}{\mathbb{F}}
 \nc{\bbG}{\mathbb{G}} \nc{\bbH}{\mathbb{H}} \nc{\bbI}{\mathbb{I}}
 \nc{\bbJ}{\mathbb{J}} \nc{\bbK}{\mathbb{K}} \nc{\bbL}{\mathbb{L}}
 \nc{\bbM}{\mathbb{M}} \nc{\bbN}{\mathbb{N}} \nc{\bbO}{\mathbb{O}}
 \nc{\bbP}{\mathbb{P}} \nc{\bbQ}{\mathbb{Q}} \nc{\bbR}{\mathbb{R}}
 \nc{\bbS}{\mathbb{S}} \nc{\bbT}{\mathbb{T}} \nc{\bbU}{\mathbb{U}}
 \nc{\bbV}{\mathbb{V}} \nc{\bbW}{\mathbb{W}} \nc{\bbX}{\mathbb{X}}
 \nc{\bbZ}{\mathbb{Z}}

 \nc{\bA}{{\bf A}} \nc{\bB}{{\bf B}} \nc{\bC}{{\bf C}}
 \nc{\bD}{{\bf D}} \nc{\bE}{{\bf E}} \nc{\bF}{{\bf F}}
 \nc{\bG}{{\bf G}} \nc{\bH}{{\bf H}} \nc{\bI}{{\bf I}}
 \nc{\bJ}{{\bf J}} \nc{\bK}{{\bf K}} \nc{\bL}{{\bf L}}
 \nc{\bM}{{\bf M}} \nc{\bN}{{\bf N}} \nc{\bO}{{\bf O}}
 \nc{\bP}{{\bf P}} \nc{\bQ}{{\bf Q}} \nc{\bR}{{\bf R}}
 \nc{\bS}{{\bf S}} \nc{\bT}{{\bf T}} \nc{\bU}{{\bf U}}
 \nc{\bV}{{\bf V}} \nc{\bW}{{\bf W}} \nc{\bX}{{\bf X}}
 \nc{\bZ}{{\bf Z}}

\nc{\cA}{{\cal A}} \nc{\cB}{{\cal B}} \nc{\cC}{{\cal C}}
\nc{\cD}{{\cal D}} \nc{\cE}{{\cal E}} \nc{\cF}{{\cal F}}
\nc{\cG}{{\cal G}} \nc{\cH}{{\cal H}} \nc{\cI}{{\cal I}}
\nc{\cJ}{{\cal J}} \nc{\cK}{{\cal K}} \nc{\cL}{{\cal L}}
\nc{\cM}{{\cal M}} \nc{\cN}{{\cal N}} \nc{\cO}{{\cal O}}
\nc{\cP}{{\cal P}} \nc{\cQ}{{\cal Q}} \nc{\cR}{{\cal R}}
\nc{\cS}{{\cal S}} \nc{\cT}{{\cal T}} \nc{\cU}{{\cal U}}
\nc{\cV}{{\cal V}} \nc{\cW}{{\cal W}} \nc{\cX}{{\cal X}}
\nc{\cZ}{{\cal Z}}

\nc{\hA}{{\hat{A}}} \nc{\hB}{{\hat{B}}} \nc{\hC}{{\hat{C}}}
\nc{\hD}{{\hat{D}}} \nc{\hE}{{\hat{E}}} \nc{\hF}{{\hat{F}}}
\nc{\hG}{{\hat{G}}} \nc{\hH}{{\hat{H}}} \nc{\hI}{{\hat{I}}}
\nc{\hJ}{{\hat{J}}} \nc{\hK}{{\hat{K}}} \nc{\hL}{{\hat{L}}}
\nc{\hM}{{\hat{M}}} \nc{\hN}{{\hat{N}}} \nc{\hO}{{\hat{O}}}
\nc{\hP}{{\hat{P}}} \nc{\hR}{{\hat{R}}} \nc{\hS}{{\hat{S}}}
\nc{\hT}{{\hat{T}}} \nc{\hU}{{\hat{U}}} \nc{\hV}{{\hat{V}}}
\nc{\hW}{{\hat{W}}} \nc{\hX}{{\hat{X}}} \nc{\hZ}{{\hat{Z}}}

\nc{\hn}{{\hat{n}}}

\def\dim{\mathop{\rm Dim}}

\def\max{\mathop{\rm max}}
\def\min{\mathop{\rm min}}

\def\tr{\mathop{\rm Tr}}

\newcommand{\bra}[1]{\langle#1|}
\newcommand{\ket}[1]{|#1\rangle}
\newcommand{\proj}[1]{| #1\rangle\!\langle #1 |}
\newcommand{\ketbra}[2]{|#1\rangle\!\langle#2|}

\newcommand{\tbc}{\red{TO BE CONTINUED...}}

\newcommand{\opp}{\red{OPEN PROBLEMS}.~}

\newcommand{\red}{\textcolor{red}}

\def\Dbar{\leavevmode\lower.6ex\hbox to 0pt
{\hskip-.23ex\accent"16\hss}D}

\begin{document}

\title{Construction of multipartite unextendible product bases and geometric measure of entanglement of  positive-partial-transpose entangled states}

\date{\today}


\author{Yize Sun}\email[]{sunyize@buaa.edu.cn}
\affiliation{LMIB and School of Mathematical Sciences, Beihang University, Beijing 100191, China}

\author{Baoshan Wang}\email[]{bwang@buaa.edu.cn}
\affiliation{LMIB and School of Mathematical Sciences, Beihang University, Beijing 100191, China}
\affiliation{International Research Institute for Multidisciplinary Science, Beihang University, Beijing 100191, China}

\author{Shiru Li}\email[]{lishiru@buaa.edu.cn}
\affiliation{LMIB and School of Mathematical Sciences, Beihang University, Beijing 100191, China}

\begin{abstract}
In quantum information theory, it is a fundamental problem to construct multipartite unextendible product bases (UPBs). We show that there exist two families UPBs in Hilbert space $\mathbb{C}^2\otimes\mathbb{C}^2\otimes\mathbb{C}^2\otimes\mathbb{C}^2\otimes\mathbb{C}^2\otimes\mathbb{C}^4$ by merging two different systems of an existing $7$-qubit UPB of size $11$. Moreover, 
a new family of $7$-qubit positive-partial-transpose (PPT) entangled
states of rank $2^7-11$ is constructed. We analytically derive a geometric measure of entanglement of a special PPT entangled states. Also an upper bound are given by two methods.
\end{abstract}


\maketitle


\section{Introduction}

The unextendible product bases (UPBs) in quantum information have been found useful in various applications \cite{1998Quantum,1999Unextendible,2003Unextendible,2012A,2011Unextendible}. Among various UPBs, the multiqubit UPBs have received extensive attentions \cite{M2012Realization}, because the multiqubit system is the mostly realizable system in experiments. For example, Bravyi and Johnston  have constructed three and four-qubit UPBs, respectively \cite{Johnston_2014,2004Unextendible}. By using the connection between UPBs and
orthogonality graphs, the four-qubit UPBs have been fully classified assisted by programmes \cite{Johnston_2014}. As far as we know, the connection between multiqubit UPBs and multipartite UPBs of higher  dimensions has been little studied. Since the latter one is harder to construct, the connection is  important for constructions. It is useful to systematically study the multipartite positive-partial-transpose (PPT) entangled states by 
using the known UPBs. Hence it is an interesting problem to construct novel PPT entangled states from multipartite UPBs.  This is the main motivation of this paper. 

Optimization theory has many applications in quantum computation \cite{2001Applying,2008Rovibrational}, quantum-state tomography \cite{2003Mode}, quantum key distribution \cite{Xing2017Implementing,2022Improving} and so on. Devetak and Winter have shown the distilling common randomness based on bipartite quantum states by using single-letter optimization \cite{2003Distilling}. To our best knowledge, it is not well-studied to apply this theory to calculate numerical value of entanglement  of quantum multipartite states. On the other hand, geometric measure can be used to compute the degree to an entangled pure quantum state, which is  characterized by the distance or angle to the nearest unentangled state \cite{2003Geometric,2002Geometric,2008Geometric}. We further consider to combine geometric measure with optimization theory to compute the degree of entanglement of multipartite PPT states. This is our second motivation.


In this paper, we apply the results of the $7$-qubit UPB of size $11$ \cite{2020The1} to construct two families UPBs in space $\mathbb{C}^2\otimes\mathbb{C}^2\otimes\mathbb{C}^2\otimes\mathbb{C}^2\otimes\mathbb{C}^2\otimes\mathbb{C}^4$. In Lemma \ref{le:2} and \ref{le:3}, we show that there exist exactly two $6$-partite UPBs by merging system $A_1A_2$ and $A_2A_3$, respectively. Based on the two lemmas, we further investigate whether the $7$-qubit UPB by merging the two remaining systems in Table \ref{tab:perfor} is a $6$-partite UPB or not. In Theorem \ref{thm:main}, we shall show that there does not exist $6$-partite UPBs except the two cases in Lemma \ref{le:2} and \ref{le:3}. Moreover, by extending to merge more systems, in Theorem \ref{thm:25}, the results present that there does not exist $n$-partite UPBs for $n<6$. 
As an application, we construct a $7$-qubit PPT entangled state of rank $2^7-11$. Based on the optimization theory of steepest descent method, Theorem \ref{thm:app1} gives a concrete numerical upper bound of the geometric measure of entanglement of both these states. Then from the viewpoint of the probability, we use Python to analyze the rationality of this upper bound. Our results shows the latest progress on the construction of multipartite UPBs and PPT entangled states. 

The structure of the paper is as follows. In Sec. \ref{sect:1}, we introduce the preliminary knowledge, which is necessary in this paper. In Sec. \ref{sec:main}, we construct two familes UPBs and show that there does not exist $n$-partite UPBs ($n<6$) constructed by $7$-qubit UPB of size $11$. In Sec. \ref{sec:app}, we establish a PPT entangled state and investigate the geometric measure of entanglement for the PPT state by using two different methods. Finally
we conclude in Sec. \ref{sect:conclu}.

\section{Preliminaries}
\label{sect:1}
In this section, we introduce the notions and facts used in this paper. In quantum mechanics, an $n$-partite quantum system with parties $A_1,\cdots,A_n$ is characterized
by a Hilbert space $\mathcal{H}=\mathcal{H}_{1}\otimes\cdots\otimes\mathcal{H}_{n}$. A product vector $\ket{\varphi}\in\mathcal{H}$ is an $n$-partite nonzero unit vector denoted as $\ket{\varphi}= \ket{\varphi_1}\otimes\cdots\otimes\ket{\varphi_n}$ with $\ket{\varphi_i}\in\mathcal{H}_{i}$. Specially, for $\dim \mathcal{H}_i=2$, it is an $n$-qubit space. For the sake of simplicity, we write it as $\ket{\varphi_1,\cdots,\varphi_n}$. Then we introduce the definition of unextendible product basis (UPB). For a set of $n$-partite orthonormal product vectors $\{\ket{\varphi_{i,1},\cdots,\varphi_{i,n}}\}$, if there is no $n$-partite product vector orthogonal
to all product vectors in the set, then it is a UPB. In particular, any orthonormal product basis in $\mathbb{C}^d$ is a UPB with size $\dim d$ and it is trivial. So we just consider the UPBs with size smaller than $\dim d$. For $7$-qubit space, we denote $\mathbb{C}^2\otimes\cdots\otimes\mathbb{C}^2$ as $\mathcal{H}_{A_1}\otimes\cdots\otimes\mathcal{H}_{A_7}$. For simplicity,  $\mathcal{H}_{A_1}\otimes\mathcal{H}_{A_2}:=\mathcal{H}_{A_1A_2}$. We take the vectors $\ket{0}:=\bma1\\0\ema$ and $\ket{1}:=\bma0\\1\ema$. Clearly, the set $\{\ket{0},\ket{1}\}$ is an orthonormal basis in $\mathbb{C}^2$. Generally, we may denote a  orthonormal basis by $\{\ket{x},\ket{x'}\}$, where $x= a,b,\cdots$. 
For a UPB, if we permutate the systems or perform any local unitary transformation, then it remains a UPB \cite{2020The}. We apply this fact in the following sections.

\section{The existence of $n$-partite UPBs constructed by $7$-qubit UPBs for $n<7$}
\label{sec:main}

We construct a $6$-partite UPB by using an existing $7$-qubit UPB of size $11$ \cite{2021The}. For this purpose we introduce the unextendible orthogonal matrices (UOMs) \cite{Chen2018Multiqubit}. We take product vectors of a $n$-partite UPB of size $m$ as row vectors of an $m\times n$ matrix, in which we shall denote orthogonal vectors $\ket{a}$ and $\ket{a'}$
by $a$ and $a'$ in UOMs and vice versa, so that the matrix is a UOM. 
Since every UPB corresponds to a UOM \cite{2020The}. We introduce the  $7$-qubit UPB of size $11$ corresponding to $11\times7$ UOM
\begin{eqnarray}
	\label{uom:1}
A =
\bma
a_{1,1}&a_{1,2}&a_{1,3}&a_{1,4}&a_{1,5}&a_{1,6}&a_{1,7}\\
a_{1,1}&a_{2,2}&a_{2,3}&a_{2,4}&a_{1,5}'&a_{2,6}&a_{2,7}\\	a_{1,1}'&a_{2,2}&a_{1,3}&a_{3,4}&a_{3,5}&a_{3,6}&a_{3,7}\\
a_{1,1}'&a_{1,2}&a_{2,3}&a_{4,4}&a_{4,5}&a_{4,6}&a_{3,7}'\\
a_{4,1}&a_{1,2}'&a_{5,3}&a_{2,4}'&a_{5,5}&a_{3,6}'&a_{5,7}\\
a_{4,1}&a_{1,2}'&a_{1,3}&a_{3,4}'&a_{6,5}&a_{2,6}'&a_{5,7}'\\
a_{4,1}'&a_{7,2}&a_{2,3}'&a_{7,4}&a_{3,5}'&a_{1,6}'&a_{7,7}\\
a_{4,1}'&a_{7,2}'&a_{2,3}'&a_{3,4}'&a_{6,5}&a_{8,6}&a_{1,7}'\\
a_{9,1}&a_{7,2}'&a_{1,3}'&a_{4,4}'&a_{5,5}'&a_{8,6}'&a_{2,7}'\\
a_{9,1}&a_{7,2}&a_{1,3}'&a_{7,4}'&a_{4,5}'&a_{2,6}'&a_{5,7}'\\
a_{9,1}'&a_{2,2}'&a_{5,3}'&a_{1,4}'&a_{6,5}'&a_{4,6}'&a_{7,7}'
\ema.
\end{eqnarray}
Suppose that the $7$-qubit UPB of size $11$ is in space $\mathcal{H}_{A_1}\otimes\cdots\otimes\mathcal{H}_{A_7}$. We merge the two systems of $A_i,i=1,\cdots,7$ such that $A$ in \eqref{uom:1} corresponds to a set of  orthonormal product vectors in space $\mathbb{C}^2\otimes\mathbb{C}^2\otimes\mathbb{C}^2\otimes\mathbb{C}^2\otimes\mathbb{C}^2\otimes\mathbb{C}^4$. For example, the merged  systems $A_1,A_2$ in \eqref{uom:1} implies that the set $\{\ket{a_{1,1},a_{1,2}}\ket{a_{1,3}}\ket{a_{1,4}}\ket{a_{1,5}}\ket{a_{1,6}}\ket{a_{1,7}},\cdots,\ket{a_{9,1}',a_{2,2}'}\ket{a_{5,3}'}\ket{a_{1,4}'}\ket{a_{6,5}'}\ket{a_{4,6}'}\ket{a_{7,7}'}\}$. There exist $21$ ways for the merged systems up to system permutation. For the convenience, we list them in Table \ref{tab:perfor}.
\renewcommand{\arraystretch}{1.5} 
\begin{table*}[htp]
		\centering
	\fontsize{8}{10}\selectfont
	\begin{threeparttable}
		\caption{The merged ways of two systems.}
		\label{tab:perfor}
		\begin{tabular}{ccccccccccc}
			\hline
			\hline
			\multirow{1}{*}{System}&$A_1$
			&$A_2$&$A_3$&$A_4$&$A_5$&$A_6$&$A_7$\cr
			\hline
			\multirow{1}{*}{$A_1$}&-
			&$A_1A_2$&$A_1A_3$&$A_1A_4$&$A_1A_5$&$A_1A_6$&$A_1A_7$\cr
			\hline
				\multirow{1}{*}{$A_2$}&-
			&-&$A_2A_3$&$A_2A_4$&$A_2A_5$&$A_2A_6$&$A_2A_7$\cr
			\hline
			\multirow{1}{*}{$A_3$}&-
			&-&-&$A_3A_4$&$A_3A_5$&$A_3A_6$&$A_3A_7$\cr
			\hline
				\multirow{1}{*}{$A_4$}&-
			&-&-&-&$A_4A_5$&$A_4A_6$&$A_4A_7$\cr
			\hline
			\multirow{1}{*}{$A_5$}&-
			&-&-&-&-&$A_5A_6$&$A_5A_7$\cr
			\hline
			\multirow{1}{*}{$A_6$}&-
			&-&-&-&-&-&$A_6A_7$\cr
			\hline
			\multirow{1}{*}{$A_7$}&-
			&-&-&-&-&-&-\cr
			\hline
			\hline
		\end{tabular}
	\end{threeparttable}
\end{table*}
Then we present the following observation.

\begin{lemma}
	\label{le:2}
When systems are merged $A_1$ and $A_2$, the $7$-qubit UPB of size $11$ is a $6$-partite UPB in this case. 
\end{lemma}
\begin{proof}
Suppose that the two merged systems are $A_1$ and  $A_2$. We may assume that the form of product vectors corresponding to $A$ in \eqref{uom:1} is 
	\begin{eqnarray}
		\label{eq:sixhcg}
		\ket{g_{j,1},\cdots,g_{j,6}}\in\mathcal{H}_{A_3}\otimes\cdots\otimes\mathcal{H}_{A_7}\otimes\mathcal{H}_{A_1 A_2},\quad j=1,\cdots,11.
	\end{eqnarray}
	Suppose that there exists a product vector $\ket{\varphi_1,\cdots,\varphi_6}$ such that it is orthogonal to all $\ket{g_{j,1},\cdots,g_{j,6}}$. For the vectors on systems $A_1,A_2$, namely the entries in the first and second columns of $A$ in \eqref{uom:1}, one can verify that the multiplicity of entries is at most two. We show an $11\times2$ submatrix with systems $A_1,A_2$ of $A$ as
	\begin{eqnarray}
		\label{eq:labelab}
		\bma
		a_{1,1}&a_{1,2}\\
		a_{1,1}&a_{2,2}\\	a_{1,1}'&a_{2,2}\\
		a_{1,1}'&a_{1,2}\\
		a_{4,1}&a_{1,2}'\\
		a_{4,1}&a_{1,2}'\\
		a_{4,1}'&a_{7,2}\\
		a_{4,1}'&a_{7,2}'\\
		a_{9,1}&a_{7,2}'\\
		a_{9,1}&a_{7,2}\\
		a_{9,1}'&a_{2,2}'
		\ema.
	\end{eqnarray}
We may assume that
\begin{eqnarray}
	\label{eq:12sincos}
\ket{a_{1,1}} &=& \bma \sin a_1\\\cos a_1\ema, \ket{a_{1,1}'} = \bma-\cos a_1\\\-\sin a_1\ema, \ket{a_{4,1}}=\bma\sin a_2\\\cos a_2\ema,\ket{a_{4,1}'}=\bma-\cos a_2\\\sin a_2\ema,\nonumber\\
\ket{a_{9,1}} &=& \bma\sin a_3\\\cos a_3\ema,\ket{a_{9,1}'}=\bma-\cos a_3\\\sin a_3\ema, \ket{a_{1,2}}=\bma \sin b_1\\\cos b_1\ema,\ket{a_{1,2}'}=\bma-\cos b_1\\\sin b_1\ema,\nonumber\\
\ket{a_{2,2}}&=&\bma \sin b_2\\\cos b_2\ema,\ket{a_{2,2}'}=\bma-\cos b_2\\\sin b_2\ema,\ket{a_{7,2}}=\bma\sin b_3\\\cos b_3\ema,\ket{a_{7,2}'}=\bma-\cos b_3\\\sin b_3\ema,
\end{eqnarray}
where $a_1,a_2,a_3,b_1,b_2,b_3\in[0,2\pi)$. 
Consider the number of rows, which is  orthogonal to $\ket{\varphi_1,\cdots,\varphi_5}\in\mathcal{H}_{A_3}\otimes\cdots\otimes\mathcal{H}_{A_7}$. In the column $4-7$ of $A$, the entries with multiplicity two are $\ket{a_{3,4}'},\ket{a_{6,5}},\ket{a_{2,6}'}$ and $\ket{a_{5,7}'}$. Moreover, $\ket{a_{3,4}'},\ket{a_{6,5}}$ are in row $6$ and $8$, and $\ket{a_{2,6}'},\ket{a_{5,7}'}$ are in row $6$ and $10$. Since the entries with multiplicity two are all in row $6$, we obtain that $\ket{\varphi_2,\varphi_3,\varphi_4,\varphi_5}$ is orthogonal at most five rows of $A$. On system $A_3$, the multiplicity of entries in column $3$ is at most three. For the entry with multiplicity three, it is $\ket{a_{1,3}}$, which is in row $1,3,6$. Since the entries with multiplicity two in column $4-7$ are all in row $6$, we obtain that $\ket{\varphi_1,\cdots,\varphi_5}$  is orthogonal to at most seven rows. So there are three cases to be considered: (a) $\ket{\varphi_1,\cdots,\varphi_5}$ is orthogonal to exactly seven rows, (b) $\ket{\varphi_1,\cdots,\varphi_5}$ is orthogonal to exactly six rows and (c) $\ket{\varphi_1,\cdots,\varphi_5}$ is orthogonal to exactly five rows.
	
	Since $\ket{\varphi_6}\in\mathbb{C}^4$ is on systems $A_1,A_2$, it  is orthogonal to at most three linearly independent vectors, which corresponding to three different rows of matrix in \eqref{eq:labelab}. If $\ket{\varphi_6}$ is orthogonal to four row vectors, then the matrix with the four rows has determinant zero. 
	Before we consider the following cases (a), (b) and (c), we only investigate that the determinant of the matrix constructed by four rows is zero in whole range $[0,2\pi)$, rather than the specific points in range $[0,2\pi)$. For example, when $\ket{\varphi_6}$ is orthogonal to row $1,3,5,11$, it implies that the determinant of $\bma a_{1,1}&a_{1,2}\\a_{1,1}'&a_{2,2}\\a_{4,1}&a_{1,2}'\\a_{4,1}'&a_{7,2}\ema$ is zero. Based on \eqref{eq:12sincos}, we obtain the determinant is 
	\begin{eqnarray}
		\label{eq:fun}
		\frac{1}{4}(\cos(2a_1-a_2-a_3)-3\cos(a_2-a_3)-2\cos(2(b_1-b_2))\sin(a_1-a_2)\sin(a_1-a_3))
	\end{eqnarray}
	Then we consider the existence of function in \eqref{eq:fun} with determinant zero. Set $a_1=2.4,a_2=1.8,a_3=0.4,b_1=4.7$, and we plot the function in \eqref{eq:fun} about $b_2$ in Fig. \ref{fig}. 
	\begin{figure}[H]
		\subfigure[]{
			\label{fig}
			\includegraphics[width=8.5cm]{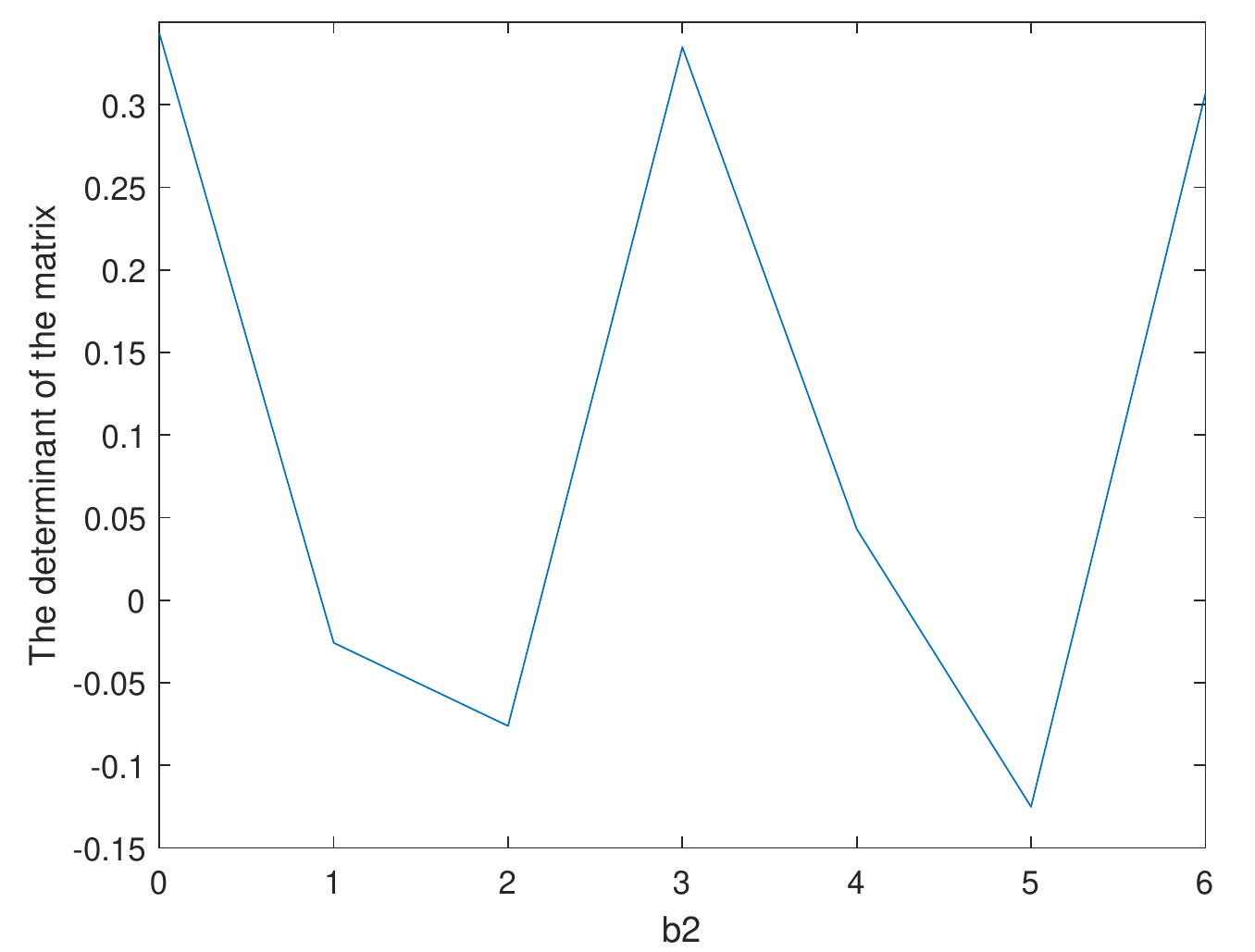}}
		\subfigure[]{
			\label{fig.b}
			\includegraphics[width=8.5cm]{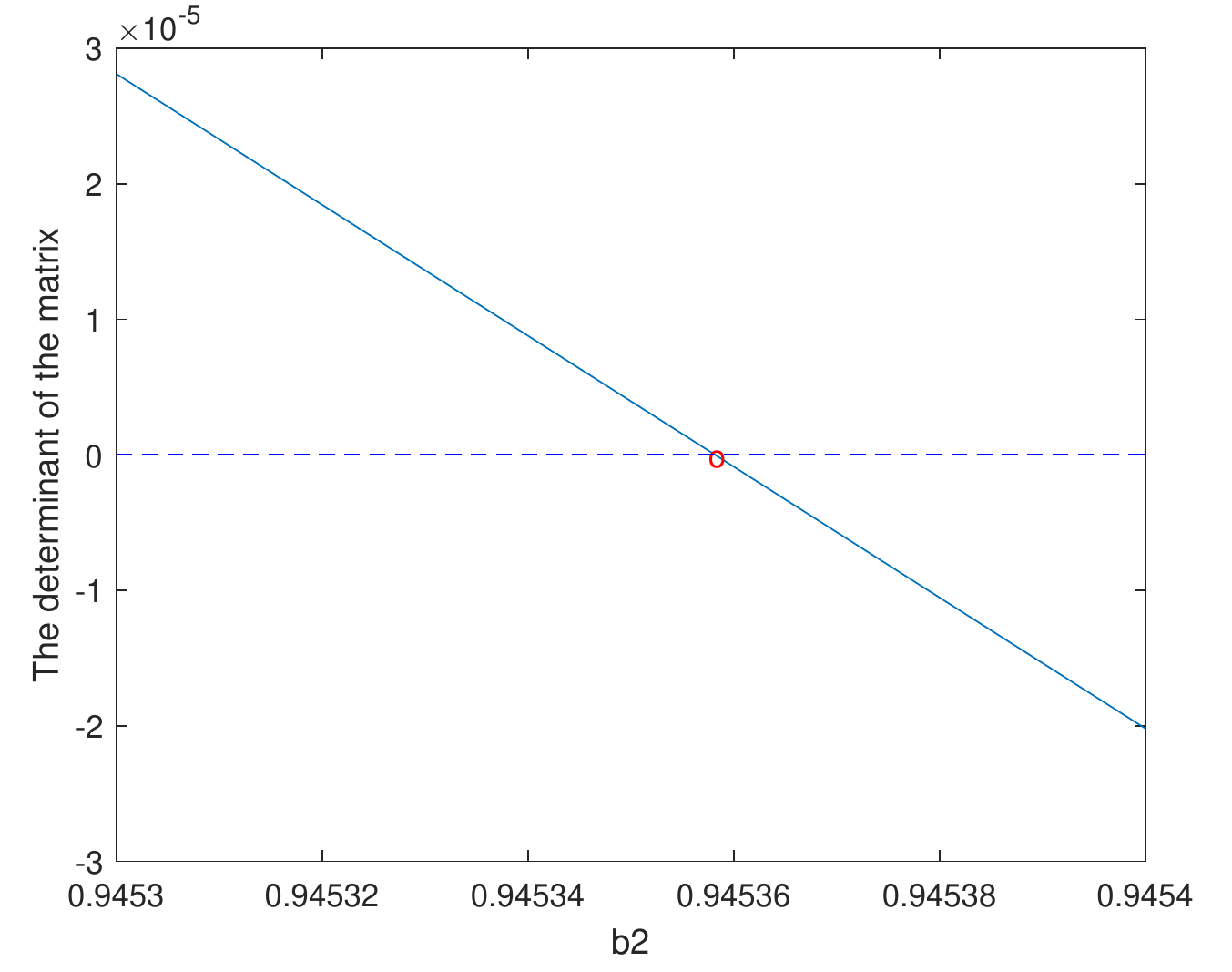}}
		\caption{Figure (a) shows the value of function in \eqref{eq:fun} and (b) shows one point of $b_2$ such that the function is zero.} \label{F5} 
	\end{figure}
	From Fig. \ref{fig}, we obtain that there exist four $b_2$ such that 
	the determinant in \eqref{eq:fun} is zero. We choose the first $b_2\in[0.9453,0.9454]$ as the solution in Fig. \ref{fig.b}. Then we obtain that there exists a $\ket{\varphi_6}$ such that it is orthogonal to four rows $1,3,5,11$. When $\ket{\varphi_1,\cdots,\varphi_5}=\ket{a_{2,3}',a_{3,4},a_{3,5},a_{8,6},a_{5,7}}$, it is orthogonal to the remaining seven rows. So we find a $\ket{\varphi_1,\cdots,\varphi_6}$ such that it is orthogonal to all row of $A$. So the $7$-qubit UPB of size $11$ is not a UPB in this case with merging systems $A_1,A_2$. However, when the determinant is not zero with some $b_2$, we obtain that it is a UPB in this case. So the results depend on $a_1,a_2,a_3,b_1,b_2,b_3$, which is uncertain. So in cases (a), (b), (c), we only consider the points in whole range $[0,2\pi)$ to make the determinant zero.
	
	(a) When $\ket{\varphi_1,\cdots,\varphi_5}$ is orthogonal to exactly seven rows, it implies that $\ket{\varphi_6}$ is orthogonal to four rows. Since $\ket{\varphi_6}\in\mathbb{C}^4$, the four rows have at most three linearly independent vectors. It implies that the determinant of  matrix constructed by the four vectors is zero. 
	
	First, we consider the cases such that $\ket{\varphi_1,\cdots,\varphi_5}$ is orthogonal to exactly seven rows. Up to systems permutation, there exist two cases about the multiplicity of entries: 
	\begin{eqnarray}
		\label{eq:31}
		3+1+1+1+1=7\quad \text{or} \quad2+2+1+1+1=7.
	\end{eqnarray}
	The entries with multiplicity two in column $4-7$ are in row $6,8$ or $6,10$. In column $3$, the entry with multiplicity three is in row $1,3,6$ and the entries with multiplicity two are in row $2,4$ or $7,8$ or $9,10$, respectively.  Suppose that the seven rows is $(x_1,x_2,x_3,x_4,x_5,x_6,x_7)$, where any two $x_i,x_j$ are distinct. For the first case in \eqref{eq:31}, since there exists exactly one entry with multiplicity three in row $1,3,6$, we have the seven rows is $(1,3,6,x_4,x_5,x_6,x_7)$. For the second case in \eqref{eq:31}, the seven rows is $(1,3,6,8,x_5,x_6,x_7), (1,3,6,10,x_5,x_6,x_7),(2,4,6,8,x_5,x_6,x_7),(2,4,6,10,x_5,x_6,x_7)$, $(7,8,6,10,x_5,x_6,x_7)$ and $(9,10,6,8,x_5,x_6,x_7)$. If there exists a product vector $\ket{\varphi_1,\cdots,\varphi_6}$ such that it is orthogonal to $A$, then $\ket{\varphi_6}$ is orthogonal to four rows without the known rows in above two cases in \eqref{eq:31}. In \eqref{eq:labelab}, if we select four arbitrary rows, then we have $330$ matrices constructed by the four rows. Specially, when the matrix has determinant zero, by using Matlab, there exist $44$ cases of four rows such that the matrices constructed by them have determinant zero. To be more intuitive, we list these cases of four rows in Table \ref{tab:1}.
	\begin{table}[htp]
		\centering
		\fontsize{8}{10}\selectfont
		\caption{The $44$ cases of four rows constructing matrices with determinant zero.}
		\begin{tabular}{cccccccc}
			\hline
			\hline
			(1,2,5,6)& (1,2,7,10)&(1,2,8,9)&(1,3,5,6)&(1,4,5,6)&(1,4,7,8)&(1,4,9,10)&(1,5,6,7)\\
			\hline
			(1,5,6,8)&(1,5,6,9)&(1,5,6,10)&(1,5,6,11)&(2,3,5,6)&(2,3,7,8)&(2,3,9,10)&(2,4,5,6)\\
			\hline
			(2,5,6,7)&(2,5,6,8)&(2,5,6,9)&(2,5,6,10)&(2,5,6,11)&(3,4,5,6)&(3,4,7,10)&(3,4,8,9)\\
			\hline
			(3,5,6,7)&(3,5,6,8)&(3,5,6,9)&(3,5,6,10)&(3,5,6,11)&(4,5,6,7)&(4,5,6,8)&(4,5,6,9)\\
			\hline
			(4,5,6,10)&(4,5,6,11)&(5,6,7,8)&(5,6,7,9)&(5,6,7,10)&(5,6,7,11)&(5,6,8,9)&(5,6,8,10)\\
			\hline
			(5,6,8,11)&(5,6,9,10)&(5,6,9,11)&(5,6,10,11)&-&-&-&-\\
			\hline
			\hline
		\end{tabular}		
		\label{tab:1}
	\end{table} 
Through analysis and observation, we obtain that these matrices have two forms by permutating systems $A_1,A_2$:
	\begin{eqnarray}
		\bma
		a&*\\
		a&*\\
		*&b\\
		*&b
		\ema,
		\quad
		\bma
		a&b\\
		a&b\\
		*&*\\
		*&*
		\ema.
	\end{eqnarray}

One can verify that there does not exist the four rows without the known rows in \eqref{eq:31}. So we obtain that when $\ket{\varphi_1\cdots,\varphi_5}$ is orthogonal to seven rows, there does not exist remaining four rows such that they are orthogonal to $\ket{\varphi_6}$. 

(b) When $\ket{\varphi_1,\cdots,\varphi_5}$ is orthogonal to exactly six rows, it implies that $\ket{\varphi_6}$ is orthogonal to five rows. 
Since $\ket{\varphi_6}\in\mathbb{C}^4$, if it is orthogonal to five rows, then any $4\times4$ matrix constructed with four rows of the five rows has determinant zero. From Table \ref{tab:1}, one can verify that there does not exist five rows such that the $4\times4$ matrices constructed by its any four rows  have determinant zero.  So we obtain that this case does not hold.

(c) When $\ket{\varphi_1,\cdots,\varphi_5}$ is orthogonal to exactly five rows, it implies that $\ket{\varphi_6}$ is orthogonal to the remaining six rows in \eqref{eq:labelab}. Since $\ket{\varphi_6}\in\mathbb{C}^4$, by using the same method in (b), we obtain that any four rows of the six rows construct matrices with determinant zero. If we regard the six rows as a $6\times4$ matrix, then it implies that it has five $5\times4$ submatrices such that any four rows of them construct $4\times4$ matrices with determinant zero. It is a contradiction with the results in (b) that there does not exist five rows such that its any four rows has determinant zero.  So this case does not hold. 

From the above cases (a), (b) and (c), we obtain that when we merge systems $A_1$ and $A_2$, the $7$-qubit UPB of size $11$ is also a $6$-partite UPB.
\end{proof}

Next, we further consider a case by merging another two systems. 

\begin{lemma}
	\label{le:3}
	When systems are merged $A_2$ and $A_3$, the $7$-qubit UPB of size $11$ is a $6$-partite UPB in this case.
\end{lemma}
\begin{proof}
We may assume that the form of product vectors corresponding to $A$ in \eqref{uom:1} is 
\begin{eqnarray}
	\label{eq:sixhcgd}
	\ket{g_{j,1},\cdots,g_{j,6}}\in\mathcal{H}_{A_1}\otimes\cdots\otimes\mathcal{H}_{A_7}\otimes\mathcal{H}_{A_2 A_3},\quad j=1,\cdots,11.
\end{eqnarray}	We show an $11\times2$ submatrix with systems $A_2,A_3$ of $A$ in \eqref{uom:1} as
\begin{eqnarray}
	\label{eq:a2a3}
	\bma
	a_{1,2}&a_{1,3}\\
a_{2,2}&a_{2,3}\\	
a_{2,2}&a_{1,3}\\
	a_{1,2}&a_{2,3}\\
a_{1,2}'&a_{5,3}\\
a_{1,2}'&a_{1,3}\\
	a_{7,2}&a_{2,3}'\\
	a_{7,2}'&a_{2,3}'\\
	a_{7,2}'&a_{1,3}'\\
	a_{7,2}&a_{1,3}'\\
	a_{2,2}'&a_{5,3}'
	\ema.
\end{eqnarray}
Suppose that there exists a product vector $\ket{\varphi_1,\cdots,\varphi_6}$, which is orthogonal to $\ket{g_{j,1},\cdots,g_{j,6}}$ in \eqref{eq:sixhcgd}. From Lemma \ref{le:2}, we obtain that on systems $A_4,\cdots,A_7$, the entries with multiplicity two are in row $6,8$ or $6,10$ and the $\ket{\varphi_2,\varphi_3,\varphi_4,\varphi_5}$ is orthogonal to
at most five rows. Since the multiplicity of entries on system $A_1$ are at most two, it implies that $\ket{\varphi_1,\cdots,\varphi_5}$ is orthogonal to at most seven rows in \eqref{eq:a2a3}. Then 
there exist three cases: (a) $\ket{\varphi_1,\cdots,\varphi_5}$ is orthogonal to exactly seven rows, (b) $\ket{\varphi_1,\cdots,\varphi_5}$ is orthogonal to exactly six rows and (c) $\ket{\varphi_1,\cdots,\varphi_5}$ is orthogonal to exactly five rows.

(a) When $\ket{\varphi_1,\cdots,\varphi_5}$ is orthogonal to exactly seven rows, it implies that $\ket{\varphi_6}$ is orthogonal to four rows.  Since $\ket{\varphi_6}\in\mathbb{C}^4$, the four rows has at most three linearly independent vectors. It implies that the determinant of  matrix is zero, which is constructed by the four vectors corresponding to the four rows. By using Matlab, there exist $23$ cases. Then we list these cases of four rows in Table \ref{tab:2}. 
\begin{table}[htp]
	\centering
	\fontsize{8}{10}\selectfont
	\caption{The $23$ cases of four rows constructing matrices with determinant zero.}
\begin{tabular}{cccccc}
	\hline
	\hline
	(1,2,3,6)& (1,3,4,6)&(1,3,5,6)&(1,3,6,7)&(1,3,6,8)&(1,3,6,9)\\
	\hline
(1,3,6,10)&(1,3,6,11)&(1,3,7,10)&(1,3,8,9)&(1,4,7,8)&(1,4,9,10)\\
\hline
(1,6,7,10)&(1,6,8,9)&(2,3,7,8)&(2,3,9,10)&(2,4,5,6)&(2,4,7,10)\\
\hline
(2,4,8,9)&(3,6,7,10)&(3,6,8,9)&(5,6,7,8)&(5,6,9,10)&-\\
	\hline
	\hline
\end{tabular}		
\label{tab:2}
\end{table}
These matrices satisfy the following  forms up to systems permutation
\begin{eqnarray}
	\bma
	a&*\\
	a&*\\
	*&b\\
	*&b
	\ema,\quad
	\bma
	a&*\\
	a&*\\
	a&*\\
	*&*
	\ema.
\end{eqnarray}
The entries in column $4-7$ with multiplicity two are in row $6,8$ or $6,10$. So the seven rows include row $6$. Then $\ket{\varphi_6}$ is orthogonal to the four rows in Table \ref{tab:2} without row $6$. There exist eight cases
\begin{eqnarray}
	\label{eq:fourzero}
	(1,3,7,10),(1,3,8,9),(1,4,7,8),(1,4,9,10),\nonumber\\
	(2,3,7,8),(2,3,9,10),(2,4,7,10),(2,4,8,9).
\end{eqnarray}
 Since $\ket{\varphi_1,\cdots,\varphi_5}$ is orthogonal to exactly seven rows, we obtain that the entries with multiplicity two on system $A_1,A_4 (A_5/A_6/A_7)$ are in four different rows. Because the entries with multiplicity two are in row $6,8$ or $6,10$, from the entries with multiplicity two on system $A_1$, the cases of four different rows are  
 \begin{eqnarray}
 	\label{eq:fourtwo}
 (1,2,6,8),(1,2,6,10),(3,4,6,8),(3,4,6,10),(7,8,6,10),(9,10,6,8).
 \end{eqnarray}
 One can verify that there does not exist the fact that any one of cases in \eqref{eq:fourzero} is in different rows with the cases in \eqref{eq:fourtwo}. So if the $\ket{\varphi_1,\cdots,\varphi_5}$ is orthogonal to seven rows, then there does not exist a $\ket{\varphi_6}\in\mathbb{C}^4$ such that it is orthogonal to the remaining four rows. 
 So (a) does not hold.

(b) When $\ket{\varphi_1,\cdots,\varphi_5}$ is orthogonal to exactly six rows, it implies that $\ket{\varphi_6}$ is orthogonal to five rows.  Since $\ket{\varphi_6}\in\mathbb{C}^4$, any four of the five rows have at most three linearly independent vectors. It implies that the determinant of  matrices is zero, which is constructed by the four vectors corresponding to the four rows. According to Table \ref{tab:2} in (a), there exist two cases $(1,3,6,7,10), (1,3,6,8,9)$ such that any four rows of the five rows corresponds to matrices with determinant zero. Since $\ket{\varphi_1,\cdots,\varphi_5}$ is orthogonal to exactly six rows, it implies that one of $\ket{\varphi_1},\cdots,\ket{\varphi_5}$ is orthogonal to an entry with multiplicity two. Since the entries on system $A_4,\cdots,A_7$ with multiplicity two are in row $6,8$ or $6,10$, if one of  $\ket{\varphi_2},\ket{\varphi_3},\ket{\varphi_4},\ket{\varphi_5}$ is orthogonal to the entry with multiplicity two, then it is orthogonal to row $6$. It shows that $\ket{\varphi_6}$ is not orthogonal to row $6$. However,  $\ket{\varphi_6}$ is orthogonal to five rows $(1,3,6,7,10)$ or $(1,3,6,8,9)$, we obtain that it does not hold. So it implies that $\ket{\varphi_1}$ on system $A_1$ is orthogonal to an entry with multiplicity two. From the first column of $A$ in \eqref{uom:1}, it shows that the entries with multiplicity two is in row $(1,2),(3,4),(5,6),(7,8)$ and $(9,10)$. Since $\ket{\varphi_6}$ is orthogonal to five rows, which are different with that of $\ket{\varphi_1}$, we obtain that there does not exist this case. So (b) does not hold.

(c) When $\ket{\varphi_1,\cdots,\varphi_5}$ is orthogonal to exactly five rows, it implies that $\ket{\varphi_6}$ is orthogonal to six rows. By using the same method in (a) and (b), we obtain that any matrices constructed by four rows of the six rows have determinant zero. For six rows, there exist six different five rows such that any four rows of these five rows have determinant zero. It is a contradiction with the fact that there exist exactly two cases $(1,3,6,7,10),(1,3,6,8,9)$ in (b). So (c) does not hold.

From above cases (a), (b) and (c),  there does not exist a product vector $\ket{\varphi_1,\cdots,\varphi_6}$ such that it is orthogonal to the $7$-qubit UPB of size $11$ in \eqref{uom:1} by merging systems $A_2$ and $A_3$. So this Lemma holds.
\end{proof}

Based on Lemma \ref{le:2} and \ref{le:3}, we have shown that the existence of $6$-partite UPBs by merging two systems $A_1,A_2$ and $A_2,A_3$, respectively. We further investigate the remaining merged ways in Table \ref{tab:perfor} and make a summary about the cases with any two  merged systems. We present the following observation. 
\begin{theorem}
	\label{thm:main}
 One of the cases by merging systems $A_1,A_2$ and $A_2,A_3$ is a $6$-partite UPB. The remaining cases by merging any two systems  are not $6$-partite UPBs. 
\end{theorem}
\begin{proof}
We prove the claim in two cases (i) the two merged systems are  without system $A_3$, (ii) the merged two systems includes system $A_3$. 

(i) When the two merged systems are without system $A_3$, there exist two cases (i.A) and (i.B) as the following.

 (i.A) The two merged systems are without system $A_1$ or $A_2$. Then we list the possible cases 
 \begin{eqnarray}
 	\label{eq:14case}
 A_1A_4,A_1A_5,A_1A_6,A_1A_7,A_2A_4,A_2A_5,A_2A_6,\nonumber\\
 A_2A_7,A_4A_5,A_4A_6,A_4A_7,A_5A_6,A_5A_7,A_6A_7.
 \end{eqnarray}
In \eqref{eq:14case},  there exist $14$ cases. We denote cases in  \eqref{eq:14case} as $\mathcal{S}_i,i=1,\cdots,14$. For cases $\mathcal{S}_i,i=1,\cdots,4,9,\cdots,14$, we may assume that the form of product vectors is 
\begin{eqnarray}
	\label{eq:sixhc}
	\ket{g^i_{j,1},\cdots,g^i_{j,6}}\in\mathcal{H}_{A_2}\otimes\cdots\otimes\mathcal{H}_{A_3}\otimes\mathcal{H}_{\mathcal{S}_i}\subseteq\mathbb{C}^2\otimes\mathbb{C}^2\otimes\mathbb{C}^2\otimes\mathbb{C}^2\otimes\mathbb{C}^2\otimes\mathbb{C}^4,
\end{eqnarray}  
where $j$ presents $j$-th row of matrix up to system permutation  of $A$. In \eqref{uom:1}, it shows that entries of the second and third columns have multiplicities at most two and three, respectively. If $A$ is not a UPB in above cases then there exists a product vector $\ket{\varphi^i_1,\cdots,\varphi^i_6}$ in space $\mathbb{C}^2\otimes\mathbb{C}^2\otimes\mathbb{C}^2\otimes\mathbb{C}^2\otimes\mathbb{C}^2\otimes\mathbb{C}^4$ such that it is orthogonal to $\ket{g^i_{j,1},\cdots,g^i_{j,6}}$. Let $\ket{g^i_{7,1}}=\ket{g^i_{10,1}}=\ket{a_{7,2}}$ and $\ket{g^i_{1,5}}=\ket{g^i_{3,5}}=\ket{g^i_{6,5}}=\ket{a_{1,3}}$ in \eqref{eq:sixhc}. When $\ket{\varphi^i_1}=\ket{a_{7,2}'},\ket{\varphi^i_5}=\ket{a_{1,3}'}$, since the multiplicities of $\ket{a_{7,2}},\ket{a_{1,3}}$ are  two and three and they are in different rows of $A$, we obtain that $\ket{\varphi^i_1,\varphi^i_5}$ is orthogonal to five rows. So we consider whether $\ket{\varphi^i_2,\varphi^i_3,\varphi^i_4,\varphi^i_6}\in\mathbb{C}^2\otimes\mathbb{C}^2\otimes\mathbb{C}^2\otimes\mathbb{C}^4$ is orthogonal to the remaining six rows on systems without $A_2,A_3$. We may assume that $\ket{\varphi^i_2},\ket{\varphi^i_3},\ket{\varphi^i_4}$ are orthogonal to $\ket{g^i_{2,2}},\ket{g^i_{4,3}},\ket{g^i_{5,4}}$,  respectively. Then we have $\ket{\varphi^i_2}=\ket{{g^{i'}_{2,2}}},\ket{\varphi^i_3}=\ket{g^{i'}_{4,3}},\ket{\varphi^i_4}=\ket{g^{i'}_{5,4}}$. Since $\ket{g^i_{2,2}},\ket{g^i_{4,3}},\ket{g^i_{5,4}}$ are in three different rows, we obtain that $\ket{\varphi^i_6}$ should be orthogonal to the remaining three rows $8,9,11$. Because $\ket{g^i_{8,6}},\ket{g^i_{9,6}},\ket{g^i_{11,6}}\in\mathbb{C}^4$, there exists a $\ket{\varphi^i_6}$ such that it is orthogonal to $\ket{g^i_{8,6}},\ket{g^i_{9,6}},\ket{g^i_{11,6}}$ on system ${\mathcal{S}_i},i=1,\cdots,4,9,\cdots,14$. Similarly, for $\mathcal{S}_i,i=5,6,7,8$, we may assume that the form of product vectors is 
\begin{eqnarray}
	\label{eq:sixhcf}
	\ket{g^i_{j,1},\cdots,g^i_{j,6}}\in\mathcal{H}_{A_1}\otimes\cdots\otimes\mathcal{H}_{A_3}\otimes\mathcal{H}_{\mathcal{S}_i}\subseteq\mathbb{C}^2\otimes\mathbb{C}^2\otimes\mathbb{C}^2\otimes\mathbb{C}^2\otimes\mathbb{C}^2\otimes\mathbb{C}^4.
\end{eqnarray}  
Let $\ket{g^i_{7,1}}=\ket{g^i_{8,1}}=\ket{a_{4,1}'}$ and $\ket{g^i_{1,5}}=\ket{g^i_{3,5}}=\ket{g^i_{6,5}}=\ket{a_{1,3}}$ in \eqref{eq:sixhc}. When $\ket{\varphi^i_1}=\ket{a_{4,1}},\ket{\varphi^i_5}=\ket{a_{1,3}'}$, since the multiplicities of $\ket{a_{7,2}},\ket{a_{1,3}}$ are  two and three and they are in different rows in $A$, we obtain that $\ket{\varphi^i_1,\varphi^i_5}$ is orthogonal to five rows. So we consider whether $\ket{\varphi^i_2,\varphi^i_3,\varphi^i_4,\varphi^i_6}\in\mathbb{C}^2\otimes\mathbb{C}^2\otimes\mathbb{C}^2\otimes\mathbb{C}^4$ is orthogonal to the remaining six rows on systems without $A_1,A_3$. We may assume that $\ket{\varphi^i_2},\ket{\varphi^i_3},\ket{\varphi^i_4}$ are orthogonal to $\ket{g^i_{2,2}},\ket{g^i_{4,3}},\ket{g^i_{5,4}}$,  respectively. Since $\ket{g^i_{2,2}},\ket{g^i_{4,3}},\ket{g^i_{5,4}}$ are in three different rows, we obtain that $\ket{\varphi^i_6}$ should be orthogonal to the remaining three rows $9,10,11$. Because $\ket{g^i_{9,6}},\ket{g^i_{10,6}},\ket{g^i_{11,6}}\in\mathbb{C}^4$, there exists a $\ket{\varphi^i_6}$ such that it is orthogonal to $\ket{g^i_{9,6}},\ket{g^i_{10,6}},\ket{g^i_{11,6}}$ on system ${\mathcal{S}_i},i=5,6,7,8$.

Hence, the $7$-qubit UPB of size $11$ is not a UPB in the above $14$ cases in \eqref{eq:14case}.
 
(i.B) When the  two merged systems are $A_1$ and $A_2$, we have proved the results in this case in Lemma \ref{le:2}.
 
From cases (i.A) and (i.B),  the $7$-qubit UPB of size $11$ is not a $6$-partite UPB in case (i.A) and is a UPB in case (i.B).

(ii) When the two merged systems includes system $A_3$, by permutating the systems and rows of $A$ in \eqref{uom:1}, we have
\begin{eqnarray}
	\label{matrix:3}
\tilde{A}=\bma
a_{1,3}&a_{1,1}&a_{1,2}&a_{1,4}&a_{1,5}&a_{1,6}&a_{1,7}\\
a_{1,3}&a_{1,1}'&a_{2,2}&a_{3,4}&a_{3,5}&a_{3,6}&a_{3,7}\\
a_{1,3}&a_{4,1}&a_{1,2}'&a_{3,4}'&a_{6,5}&a_{2,6}'&a_{5,7}'\\
a_{1,3}'&a_{9,1}&a_{7,2}'&a_{4,4}'&a_{5,5}'&a_{8,6}'&a_{2,7}'\\
a_{1,3}'&a_{9,1}&a_{7,2}&a_{7,4}'&a_{4,5}'&a_{2,6}'&a_{5,7}'\\
a_{2,3}&a_{1,1}'&a_{1,2}&a_{4,4}&a_{4,5}&a_{4,6}&a_{3,7}'\\
a_{2,3}&a_{1,1}&a_{2,2}&a_{2,4}&a_{1,5}'&a_{2,6}&a_{2,7}\\
a_{2,3}'&a_{4,1}'&a_{7,2}&a_{7,4}&a_{3,5}'&a_{1,6}'&a_{7,7}\\
a_{2,3}'&a_{4,1}'&a_{7,2}'&a_{3,4}'&a_{6,5}&a_{8,6}&a_{1,7}'\\
a_{5,3}&a_{4,1}&a_{1,2}'&a_{2,4}'&a_{5,5}&a_{3,6}'&a_{5,7}\\
a_{5,3}'&a_{9,1}'&a_{2,2}'&a_{1,4}'&a_{6,5}'&a_{4,6}'&a_{7,7}'
\ema.
\end{eqnarray}
Suppose that there exists a product vector $\ket{\phi_1,\cdots,\phi_6}$ on systems $A_l,\cdots,A_k,A_3A_i$ such that it is orthogonal to all rows of $\tilde{A}$ in \eqref{matrix:3}, where $i\in\{1,\cdots,7\}, l,k\in\{1,\cdots,7\}\backslash\{3,i\}$ and $l< k$. 

If we merge the systems $A_3$ and $A_1$, then we obtain that $\ket{a_{1,3}',a_{9,1}}$ and $\ket{a_{2,3}',a_{4,1}'}$ have  multiplicities two in $\tilde{A}$, respectively. Since $\ket{\phi_6}\in\mathbb{C}^4$,  $\ket{\phi_6}$ is orthogonal to five rows of $\tilde{A}$ on system $A_3A_1$. If $\ket{\phi_6}$ is orthogonal to five rows of $\tilde{A}$ with  $\ket{a_{1,3}',a_{9,1}}$ on row $4,5$ and $\ket{a_{2,3}',a_{4,1}'}$ on row $8,9$, then $\ket{\phi_1,\cdots,\phi_5}$ is orthogonal to remaining six rows of a $7\times5$ matrix
$\bma
a_{1,2}&a_{1,4}&a_{1,5}&a_{1,6}&a_{1,7}\\
a_{2,2}&a_{3,4}&a_{3,5}&a_{3,6}&a_{3,7}\\
a_{1,2}'&a_{3,4}'&a_{6,5}&a_{2,6}'&a_{5,7}'\\
a_{1,2}&a_{4,4}&a_{4,5}&a_{4,6}&a_{3,7}'\\
a_{2,2}&a_{2,4}&a_{1,5}'&a_{2,6}&a_{2,7}\\
a_{1,2}'&a_{2,4}'&a_{5,5}&a_{3,6}'&a_{5,7}\\
a_{2,2}'&a_{1,4}'&a_{6,5}'&a_{4,6}'&a_{7,7}'
\ema$. In this $7\times5$ matrix, it shows that the entries with multiplicity two are $\ket{a_{1,2}}, \ket{a_{2,2}},\ket{a_{1,2}'}$. Since they are in the same column and the remaining entries have multiplicity one, we obtain that $\ket{\phi_1,\cdots,\phi_5}$ is orthogonal to at most six rows in the $7\times5$ matrix. So we find  a product vector $\ket{a_{1,2}',a_{3,4}',a_{6,5}',a_{2,6}',a_{5,7}', \phi_6}$ such that it is orthogonal to all rows of $\tilde{A}$ in \eqref{matrix:3}, where $\ket{\phi_6}\in\mathbb{C}^4$ is orthogonal to $\ket{a_{1,3}',a_{9,1}},\ket{a_{2,3}',a_{4,1}'}$ and $\ket{a_{5,3}',a_{9,1}'}$. 

Next, we consider the two merged systems are $A_3$ and other systems without $A_1$. Namely, the two merged systems are  $A_3,A_i,i\in\{2,4,5,6,7\}$. 
If $i=2$, we have proved  this case  in Lemma \ref{le:3}. 

If $i\in\{4,5,6,7\}$, since the entries with multiplicity two are in the same row, then $\ket{\phi_3,\phi_4,\phi_5}$ is orthogonal to at most four rows. Because the multiplicity of the entries on system $A_1,A_2$ is at most two,  $\ket{\phi_1,\phi_2}$ is orthogonal to at most four rows. So it shows that $\ket{\phi_1\cdots,\phi_5}$ is orthogonal to at most eight rows. 
If $\ket{\phi_1\cdots,\phi_5}$ is orthogonal to exactly eight rows, then $\ket{\phi_3,\phi_4,\phi_5}$ is orthogonal to row $3$ of $\tilde{A}$ in \eqref{matrix:3}. Since $\ket{\phi_1,\phi_2}$ is orthogonal to entries with multiplicity two, we obtain that  $\ket{\phi_1,\phi_2}$ is orthogonal to four rows of an $8\times2$ matrix
$
\bma
a_{1,1}&a_{1,2}\\
a_{1,1}'&a_{2,2}\\
a_{9,1}&a_{7,2}'\\
a_{9,1}&a_{7,2}\\
a_{1,1}'&a_{1,2}\\
a_{1,1}&a_{2,2}\\
a_{4,1}'&a_{7,2}\\
a_{4,1}'&a_{7,2}'\\
\ema
$ of $\tilde{A}$. By observing the $8\times2$ matrix, there exist the entries in column $1$ and $2$ with multiplicity two, which are in four different rows.  In $\tilde{A}$, the entries $\ket{a_{3,4}'}, \ket{a_{6,5}},\ket{a_{2,6}'},\ket{a_{5,7}'}$ on systems $A_4,\cdots,A_7$ with multiplicity two are in row $3, 9$ and $3,5$,  respectively. If $\ket{\phi_3}$ is orthogonal to row $3,9$, when $i=4$, then there exists a product vector $\ket{a_{1,1}',a_{7,2}',a_{6,5}',a_{3,6}',a_{2,7},\phi_6}$, which is orthogonal to all rows of $\tilde{A}$ in \eqref{matrix:3}, where $\ket{\phi_6}\in\mathbb{C}^4$ is orthogonal to $\ket{a_{2,3},a_{4,4}},\ket{a_{5,3},a_{2,4}'}$ and $\ket{a_{5,3}',a_{1,4}'}$. Similarly, when $i=5$, there exists a product vector $\ket{a_{1,1}',a_{7,2}',a_{3,4}, a_{3,6}',a_{2,7},\phi_6}$, which is orthogonal to all rows of $\tilde{A}$ in \eqref{matrix:3}, where $\ket{\phi_6}\in\mathbb{C}^4$ is orthogonal to $\ket{a_{2,3},a_{4,5}},\ket{a_{5,3},a_{5,5}}$ and $\ket{a_{5,3}',a_{6,5}'}$. When $i=6$, there exists a product vector $\ket{a_{1,1}',a_{7,2}',a_{3,4}, a_{3,5}',a_{2,7},\phi_6}$ is orthogonal to all rows of $\tilde{A}$ in \eqref{matrix:3}, where $\ket{\phi_6}\in\mathbb{C}^4$ is orthogonal to $\ket{a_{2,3},a_{4,6}},\ket{a_{5,3},a_{3,6}'}$ and $\ket{a_{5,3}',a_{4,6}'}$.  When $i=7$, there exists a product vector $\ket{a_{1,1}',a_{7,2}',a_{3,4}, a_{3,5}',a_{8,6},\phi_6}$ is orthogonal to all rows of $\tilde{A}$ in \eqref{matrix:3}, where $\ket{\phi_6}\in\mathbb{C}^4$ is orthogonal to $\ket{a_{2,3},a_{3,7}'},\ket{a_{5,3},a_{5,7}}$ and $\ket{a_{5,3}',a_{7,7}'}$. From above cases in (ii), we obtain that $7$-qubit UPB of size $11$ is a UPB in six partitions with  merged systems $A_2,A_3$ and is not a UPB in other merged ways.  

Besed on the above cases (i) and (ii), when the two merged systems are $A_1,A_2$ and $A_2,A_3$, the $7$-qubit UPB of size $11$ is a UPB in six partitions. Otherwise, it is not a $6$-partite UPB.
\end{proof}

Theorem \ref{thm:main} shows the existence of $6$-partite UPBs constructed by the $7$-qubit UPB of size $11$. By extending to merge more systems, we further discuss the existence of $n$-partite UPBs for $2\leq n\leq5$ in the following theorem.

\begin{theorem}
	\label{thm:25}
	The $7$-qubit UPB of size $11$ is not an $n$-partite UPB for $2\leq n\leq 5$.
\end{theorem}
We give the proof of this theorem in Appendix \ref{app:thm25}.
From  Thoerem \ref{thm:main} and \ref{thm:25}, we have investigated the $7$-qubit UPB of size $11$ in $n$-partitions, where $2\leq n\leq6$. The results show that there exist exactly two cases by merging systems $A_1,A_2$ and $A_2,A_3$. During the analysis, by using the way of combination, we have shown the existence of multipartite UPBs of higher dimension with the parameters in the whole range $[0,2\pi)$.

\section{The construction and entanglement of $7$-qubit positive-partal-transpose entangled states}
\label{sec:app}
In this section, we apply a concrete $7$-qubit UPB of size $11$ to construct multiqubit entangled states, which have positive partial transpose (PPT). Given a bipartite state $\rho\in\mathcal{B}(\mathbb{C}^m\otimes\mathbb{C}^n)$ with the normalization condition $\tr\rho=1$, the partial transpose of $\rho$ is defined as $\rho^\Gamma=\sum_{ij}\ketbra{a_j}{a_i}\otimes\bra{a_i}\rho\ket{a_j}$, where $\{\ket{a_i}\}$ is a basis in $\mathcal{H}_A$. If $\rho^\Gamma$ is positive semidefinite, then we say that $\rho$ is positive partial transpose (PPT). Otherwise, $\rho$ is non-PPT (NPT) \cite{2020The1}. Moreover, the construction of multipartite PPT entangled states and its study are involved a lot in entanglement theory \cite{1998Positive,2007Class,2013Geometry}. A systematic method for such a construction is to employ an $n$-partite UPB $\{\ket{\phi_i}\}_{i=1,\cdots,m}\in\mathbb{C}^{d_1}\otimes\cdots\otimes\mathbb{C}^{d_n}$. We can show that 
\begin{eqnarray}
	\label{eq:pptphi}
	\rho&=&\frac{1}{d_1\cdots d_n-m}(\mathbb{I}_{d_1\cdots d_n}-\sum_{i}\proj{\phi_i})
\end{eqnarray}
is an $n$-partite PPT entangled state of rank $d_1\cdots d_n-m$. 
We use this method to construct a $7$-partite PPT entangled state $\alpha$ by using a concrete $7$-qubit UPB $\{\ket{\varphi_i},i=1,\cdots,11\}$ of size $11$. From \eqref{eq:pptphi}, we have 
\begin{eqnarray}
	\label{eq:geal}
	\alpha&=&\frac{1}{2^7-11}(\mathbb{I}_{2^7}-\sum_{i=1}^{11}\proj{\varphi_i})\nonumber\\
		&=& \frac{1}{2^7-11}(\mathbb{I}_{2^7}-Q),
\end{eqnarray}
where $Q=\sum_{i=1}^{11}\proj{\varphi_i}$. 
We list $\{\ket{\varphi_i}\}$ as follows: 
\begin{eqnarray}
	\label{eq:phi17}
	\ket{\varphi_1}&=& \ket{0,0,0,0,0,0,0},\nonumber\\
	\ket{\varphi_2}&=& \ket{0,\frac{1}{\sqrt{2}}(0+1),\frac{1}{\sqrt{2}}(0+1),\frac{1}{\sqrt{2}}(0+1),1,\frac{1}{\sqrt{2}}(0+1),\frac{1}{\sqrt{2}}(0+1)},\nonumber\\
	\ket{\varphi_3}&=& \ket{1,\frac{1}{\sqrt{2}}(0+1),0,\frac{1}{\sqrt{3}}(0+\sqrt{2}\cdot1),\frac{1}{\sqrt{2}}(0+1),\frac{1}{\sqrt{3}}(0+\sqrt{2}\cdot1),\frac{1}{\sqrt{3}}(0+\sqrt{2}\cdot1)},\nonumber\\
	\ket{\varphi_4}&=&\ket{1,0,\frac{1}{\sqrt{2}}(0+1),\frac{1}{2}(0+\sqrt{3}\cdot1),\frac{1}{\sqrt{3}}(0+\sqrt{2}\cdot1),\frac{1}{2}(0+\sqrt{3}\cdot1),\frac{1}{\sqrt{3}}(\sqrt{2}\cdot0-1)},\nonumber\\
	\ket{\varphi_5}&=& \ket{\frac{1}{\sqrt{2}}(0+1),1,\frac{1}{\sqrt{3}}(0+\sqrt{2}\cdot1),\frac{1}{\sqrt{2}}(0-1),\frac{1}{{2}}(0+\sqrt{3}\cdot1),\frac{1}{\sqrt{3}}(\sqrt{2}\cdot0-1),\frac{1}{{2}}(0+\sqrt{3}\cdot1)},\nonumber\\
	\ket{\varphi_6}&=&\ket{\frac{1}{\sqrt{2}}(0+1),1,0,\frac{1}{\sqrt{3}}(\sqrt{2}\cdot0-1),\frac{1}{\sqrt{5}}(0+2\cdot1),\frac{1}{\sqrt{2}}(0-1),\frac{1}{2}(\sqrt{3}\cdot0-1)},\nonumber\\
	\ket{\varphi_7}&=& \ket{\frac{1}{\sqrt{2}}(0-1),\frac{1}{\sqrt{3}}(0+\sqrt{2}\cdot1),\frac{1}{\sqrt{2}}(0-1),\frac{1}{\sqrt{5}}(0+2\cdot1),\frac{1}{\sqrt{2}}(0-1),1,\frac{1}{\sqrt{5}}(0+2\cdot1)},\nonumber\\
	\ket{\varphi_8}&=&\ket{\frac{1}{\sqrt{2}}(0-1),\frac{1}{\sqrt{3}}(\sqrt{2}\cdot0-1),\frac{1}{\sqrt{2}}(0-1),\frac{1}{\sqrt{3}}(\sqrt{2}\cdot0-1),\frac{1}{\sqrt{5}}(0+2\cdot1),\frac{1}{\sqrt{5}}(0+2\cdot1),1},\nonumber\\
	\ket{\varphi_9}&=& \ket{\frac{1}{\sqrt{3}}(0+\sqrt{2}\cdot1),\frac{1}{\sqrt{3}}(\sqrt{2}\cdot0-1),1,\frac{1}{2}(\sqrt{3}\cdot0-1),\frac{1}{2}(\sqrt{3}\cdot0-1),\frac{1}{\sqrt{5}}(\sqrt{2}\cdot0-1),\frac{1}{\sqrt{2}}(0-1)},\nonumber\\
	\ket{\varphi_{10}}&=&\ket{\frac{1}{\sqrt{3}}(0+\sqrt{2}\cdot1),\frac{1}{\sqrt{3}}(0+\sqrt{2}\cdot1),1,\frac{1}{\sqrt{5}}(2\cdot0-1),\frac{1}{\sqrt{3}}(\sqrt{2}\cdot0-1),\frac{1}{\sqrt{2}}(0-1),\frac{1}{2}(\sqrt{3}\cdot0-1)},\nonumber\\
	\ket{\varphi_{11}}&=& \ket{\frac{1}{\sqrt{3}}(\sqrt{2}\cdot0-1),\frac{1}{\sqrt{2}}(0-1),\frac{1}{\sqrt{3}}(\sqrt{2}\cdot0-1),1,\frac{1}{\sqrt{5}}(\sqrt{2}\cdot0-1),\frac{1}{2}(\sqrt{3}\cdot0-1),\frac{1}{\sqrt{5}}(\sqrt{2}\cdot0-1)}.
\end{eqnarray}

Before investigating the geometric measure of entanglement of $\alpha$, we introduce the definition of geometric measure \cite{2003Geometric,2020The}. 
For any $n$-partite quantum state $\rho$, the measure is defined as
\begin{eqnarray}
\label{eq:gm}
G(\rho):=-\log_2\max_{\ket{\delta_1},\cdots,\ket{\delta_n}}\bra{\delta_1,\cdots,\delta_n}\rho\ket{\delta_1,\cdots,\delta_n},
\end{eqnarray}
where $\ket{\delta_1,\cdots,\delta_n}$ is a normalized product state in $\mathbb{C}^{d_1}\otimes\cdots\otimes\mathbb{C}^{d_n}$. If a $6$-partite state $\sigma\in \mathcal{B}(\mathbb{C}^2\otimes\mathbb{C}^2\otimes\mathbb{C}^2\otimes\mathbb{C}^2\otimes\mathbb{C}^2\otimes\mathbb{C}^4)$, then it implies $d_1=d_2=d_3=d_4=d_5=2$ and $d_6=4$ with $n=6$. Similarly,  it implies that  state  $\alpha\in\mathcal{B}(\mathbb{C}^2\otimes\cdots\otimes\mathbb{C}^2)$ in \eqref{eq:geal} has $d_1=\cdots=d_7=2$. The definition of geometric measure
of entanglement in \eqref{eq:gm} shows that the more systems we combine together, the less
entanglement we obtain. Then we have  
\begin{eqnarray}
	\label{eq:18}
	G(\sigma)\leq G(\alpha).
\end{eqnarray}
So $G(\alpha)$ is the upper bound of $G(\sigma)$, which may be not a tight upper bound. From \eqref{eq:geal} and \eqref{eq:gm}, 
we have 
\begin{eqnarray}
	\label{eq:maxalpha}
	G(\alpha)= -\log_2\max_{\ket{\delta_1},\cdots,\ket{\delta_7}}\bra{\delta_1,\cdots,\delta_7}\alpha\ket{\delta_1,\cdots,\delta_7}. 
\end{eqnarray} 
The evaluation of $G(\alpha)$ is equivalent to 
\begin{eqnarray}
	\label{eq:min}
	\min_{\ket{\delta_1},\cdots,\ket{\delta_7}\in\mathbb{C}^2\otimes\cdots\otimes\mathbb{C}^2}\bra{\delta_1,\cdots,\delta_7}Q\ket{\delta_1,\cdots,\delta_7}
\end{eqnarray}
for $Q\in\mathcal{B}(\mathbb{C}^2\otimes\cdots\otimes\mathbb{C}^2)$ in \eqref{eq:geal}. 
To investigate the minimum of $G(\alpha)$, we consider $\ket{\delta_1,\cdots,\delta_7}\in\mathbb{R}^2\otimes\cdots\otimes\mathbb{R}^2$, and assume 
\begin{eqnarray}
	\label{eq:sixdelta}
\ket{\delta_1}
=
\bma
\sin \upsilon_1\\
\cos\upsilon_1
\ema,
\quad
\ket{\delta_2}
&=&
\bma
\sin\upsilon_2\\
\cos\upsilon_2
\ema,
\quad
\ket{\delta_3}
=
\bma
\sin \upsilon_3\\
\cos\upsilon_3
\ema,
\quad
\ket{\delta_4}
=
\bma
\sin \upsilon_4\\
\cos\upsilon_4
\ema,
\nonumber
\end{eqnarray}
\begin{eqnarray}
	\label{eq:delta7}
\ket{\delta_5}
=
\bma
\sin \upsilon_5\\
\cos\upsilon_5
\ema
\quad,
\ket{\delta_6}
=
\bma
\sin \upsilon_6\\
\cos\upsilon_6
\ema,
\quad
\ket{\delta_7}
=
\bma
\sin \upsilon_7\\
\cos\upsilon_7
\ema,
\end{eqnarray}
where $\upsilon_i\in[0,2\pi),i=1,\cdots,7$. So \eqref{eq:min} becomes 
\begin{eqnarray}
	\label{eq:rfun}
	\label{eq:min2}
	\min_{\ket{\delta_1},\cdots,\ket{\delta_7}\in\mathbb{R}^2\otimes\cdots\otimes\mathbb{R}^2}\bra{\delta_1,\cdots,\delta_7}Q\ket{\delta_1,\cdots,\delta_7}.
\end{eqnarray}
By substituting $\ket{\delta_i}$ in  \eqref{eq:delta7} to the target function in \eqref{eq:rfun}, we obtain that it is a continuous function. Due to the differentiability of function in \eqref{eq:rfun} and its complex Hessian matrix,  we adopt classical one-order optimization method: steepest descent method. This method is simple and efficient {for this problem.} We select the initial point $(\upsilon_1,\cdots,\upsilon_7)=(0,0,0,0,0,0,0)$ and set termination criteria, i.e. $\|g{(x_k)}\|_2\leq 10^{-4}$. For the convenience of readers, we show the algorithm in Alg. \ref{SD}. 

\begin{algorithm}[htb] 
	\caption{Steepest Descent Method} 
	\label{SD}  Solve problem \textbf{function f$(x_k)$} = 	$\min_{\ket{\delta_1},\cdots,\ket{\delta_7}\in\mathbb{C}^2\otimes\cdots\otimes\mathbb{C}^2}\bra{\delta_1,\cdots,\delta_7}Q\ket{\delta_1,\cdots,\delta_7}$;
	\begin{algorithmic}[1] 
		\REQUIRE ~~\\ 
		Initial $x_0=(\upsilon_1,\cdots,\upsilon_7)= (0,\cdots,0)$, 
		$a_0=10$;\\
		\ENSURE ~~\\ 
		$x_{k+1}$ and $\textbf{f}(x_{k+1})$;
		\FOR{$k = 1,\cdots,n$} 
		\label{ code:fram:extract }
		\STATE Derivative the parameter $\upsilon_1,\cdots,\upsilon_7$ of $\ket{\delta_1},\cdots,\ket{\delta_7}$ respectively;\ 
		\label{code:fram:trainbase}
		\STATE Reassign the seven $128\times1$ vectors in a $128\times7$ matrix $M(x_k)$ and do transpose;\ 
		\label{code:fram:add}
		\STATE Let gradient $g(x_k)$ be $M_{x_k}^TQx_k$; \ 
		\STATE $x_{k+1} = x_k-a_kg(x_k)$;\ 
		\STATE Find an $a_{k+1}>0$, i.e. $\textbf{f} (x_{k+1}) < \textbf{f}(x_k)+g^T(x_k)(x_{k+1}-x_k)+\frac{a_{k+1}}{2}\|x_{k+1}-x_k\|_2^2$;
		\ENDFOR
		\STATE \textbf{stopping criterion}: $\|g(x_{k+1})\|_2\leq 10^{-4}$;\
	\end{algorithmic}
\end{algorithm}

The algorithm iterates $121$ times to satisfy the termination criterion. The optimal solution in \eqref{eq:rfun} (also in \eqref{eq:maxalpha}) is about $3.18624\times10^{-5}$, and the optimal value of the function in \eqref{eq:maxalpha} is $6.87041$ ebits. The iteration curve of $G(\alpha)$ is shown in Fig. \ref{fig:matlabmin}. Further, the parameter $(\upsilon_1,\cdots,\upsilon_7)$ of $\ket{\delta_1},\cdots,\ket{\delta_7}$ is $(2.35414, 2.83365, 3.14800, 0.615284, 3.92691, 2.35162, 3.61857)$. 

\begin{figure}[htp]
	\centering
	\includegraphics[scale=0.6,angle=0]{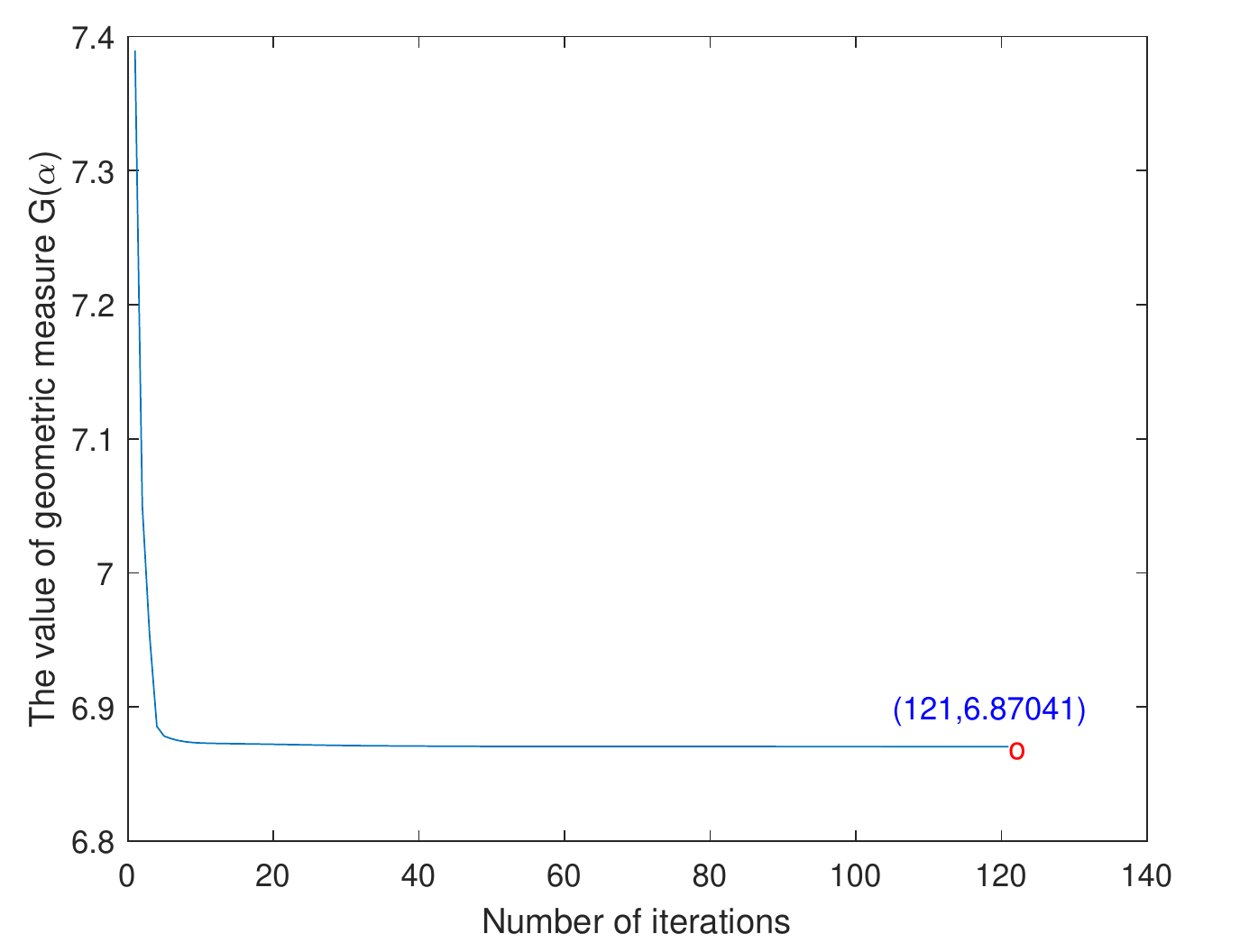}
	\caption{The blue line shows that the entanglement of the PPT state $\alpha$ by using the geometric measure. It implies when the number of iterations is $121$, the entanglement converges to $6.87041$ ebits.} 
	\label{fig:matlabmin}
\end{figure}

{One can verify that the Hessian matrix is quite complex when $\upsilon_1=\cdots=\upsilon_7=\pi/2$, so the  objective function in \eqref{eq:rfun} is not convex. Above algorithm in  Alg. \ref{SD} do not guarantee that the parameter $(2.35414, 2.83365, 3.14800, 0.615284, 3.92691, 2.35162, 3.61857)$ is a global solution. } 
We adopt uniform random sampling to illustrate the reasonable of this solution. Elect ten millon samples by using  Python to verify it in Fig. \ref{fig:pythonmin1}.  
\begin{figure}[htp]
	\centering
	\includegraphics[scale=0.6,angle=0]{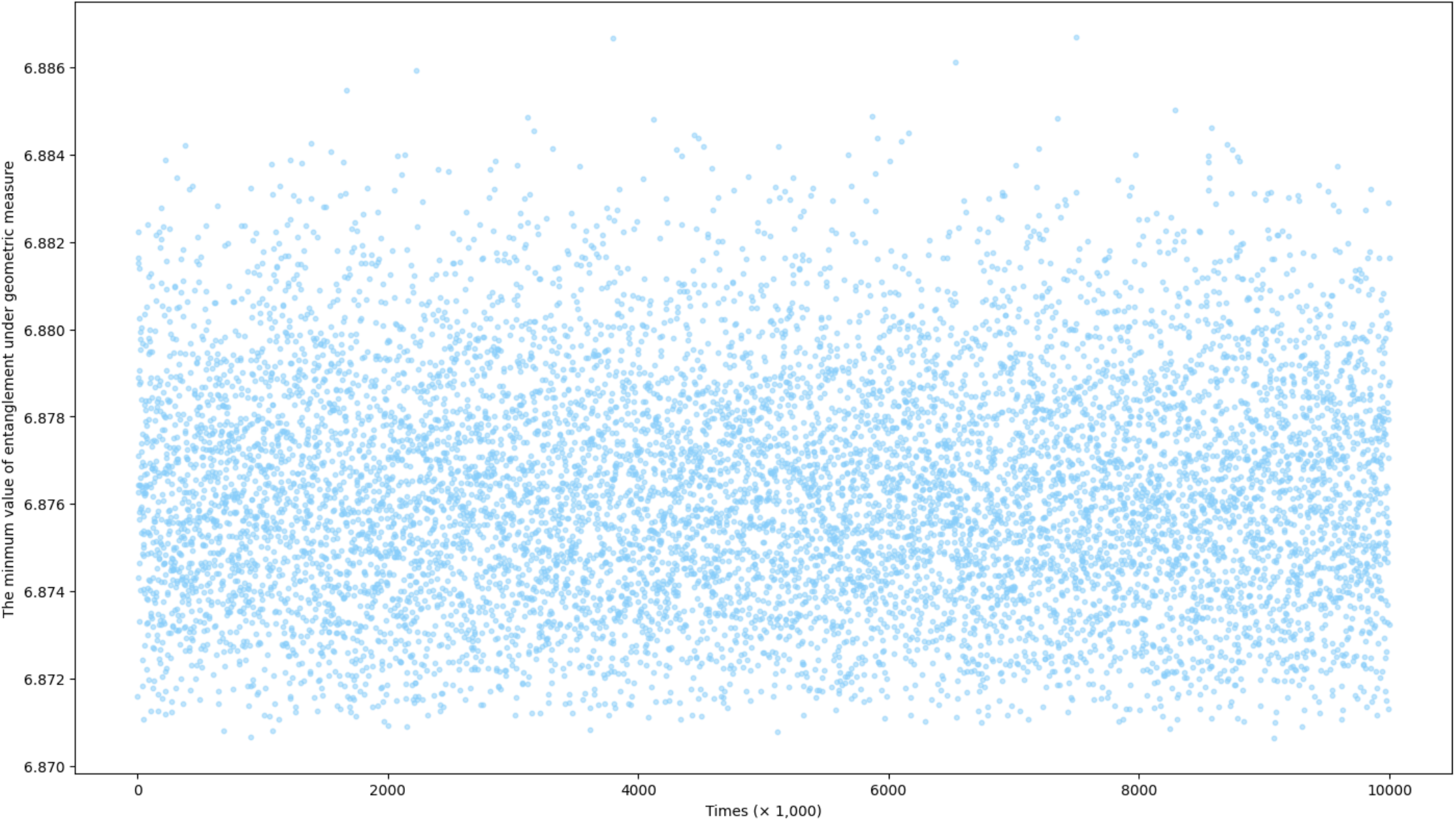}
	\caption{This figure shows that the  ten million entanglement values in \eqref{eq:gm} by using geometric measure with random parameters $\upsilon_1,\cdots,\upsilon_7$ in range $[0,2\pi)$.} 
	\label{fig:pythonmin1}
\end{figure}
Then we obtain that the entanglement is $6.87061$ ebits in \eqref{eq:gm}. Further, the parameter $(\upsilon_1,\cdots,\upsilon_7)$ of $\ket{\delta_1},\cdots,\ket{\delta_7}$ is
$(0.81208, 0.03589, 0.01666, 0.12369, 1.55458, 0.60677, 0.52988)$. Compared with the above results, we obtain that the error is $0.00291$$\%$. Due to the randomness of this method by using Python, we reasonably obtain that optimal value $6.87041$ of function in \eqref{eq:maxalpha} is a global minimum. So it is an upper bound. To conclude, we summarize this result in Theorem \ref{thm:app1}.

\begin{theorem}
	\label{thm:app1}
	For the $7$-qubit PPT state $\alpha$ in Eq. \eqref{eq:geal}, 
	its geometric measure of entanglement is $G(\alpha)\approx6.87041$ ebits with parameter  $(v_1,\cdots,v_7)=(2.35414, 2.83365, 3.14800, 0.615284, 3.92691, 2.35162, 3.61857)$.
\end{theorem}

It is the main result in this section. Based on the results in  Theorem \ref{thm:main} and \ref{thm:25}, from \eqref{eq:phi17}, by merging system $A_1, A_2$ or $A_2,A_3$, we can construct $6$-partite PPT entangled states by using the two $6$-partite UPBs of size $11$, respectively. From \eqref{eq:18}, we obtain that  the range of their geometric measure of entanglement is $(0,6.87041]$. The results have shown the numerical value of entanglement of  multipartite PPT states of higher dimensions. In addition, we obtain the upper bound by using the optimization method and probabilistic methods, which supply new ways to compute the entanglement value of multipartite entangled states.  





\section{Conclusions}
\label{sect:conclu}

In this paper, we have applied $7$-qubit UPB of size $11$ to construct more UPBs by merging systems. We have shown that there are exactly two UPBs in space $\mathbb{C}^2\otimes\mathbb{C}^2\otimes\mathbb{C}^2\otimes\mathbb{C}^2\otimes\mathbb{C}^2\otimes\mathbb{C}^4$. No UPBs are in two, three, four, five partitions.  Further, we have worked out the entanglement of a concrete $7$-qubit PPT entangled state using the geometric measure of entanglement and given a concrete value.

One problem is to construct more multipartite UPBs using the existing multiqubit UPBs. Moreover, it is a challenge to find a general method to solve $n$-partite UPBs and consider its applications in various information tasks. An another direction is to further investigate an entanglement value such that the $6$-partite PPT entangled state is genuinely entangled. 

\section*{Disclosure statement}
No potential conflict of interest was reported by the authors.

\section*{Acknowledgments}
Authors were supported by NSFC (Grant No. 11871083).

\appendix

\section{The proof of Theorem \ref{thm:25}}
\label{app:thm25}
We prove this theorem in the four following lemma \ref{le:five} - \ref{le:two}, which show that the $7$-qubit UPB of size $11$ is not an $n$-partite UPB for $2\leq n \leq 5$. 
\begin{lemma}
	\label{le:five}
	The $7$-qubit UPB of size $11$ is not a UPB in any five partitions.
\end{lemma}
\begin{proof}
	We denote $(s,x,y,z,v)$ as the number of the merged systems  in five partitions. Since the $7$-qubit UPB of size $11$ is on seven systems, up to systems permutation, we obtain two cases $(1,1,1,1,3)$ and $(1,1,1,2,2)$.
	
	For the first case $(1,1,1,1,3)$, it implies that the product vectors of $A$  in \eqref{uom:1} are in space $\mathbb{C}^2\otimes\mathbb{C}^2\otimes\mathbb{C}^2\otimes\mathbb{C}^2\otimes\mathbb{C}^{2^3}$. Up to systems permutations of $A$, we may assume that the form of product vectors of $A$ is  
	\begin{eqnarray}
		\label{eq:11113}
		\ket{g_{i,1},g_{i,2},g_{i,3},g_{i,4},g_{i,5}}\in\mathbb{C}^2\otimes\mathbb{C}^2\otimes\mathbb{C}^2\otimes\mathbb{C}^2\otimes\mathbb{C}^{2^3},
	\end{eqnarray}
	where $i=1,\cdots,11$. Suppose that there exists a product vector $\ket{\phi_1,\cdots,\phi_5}$ such that it is orthogonal to all rows of $A$ up to systems  permutation. Since $\ket{\phi_5}\in\mathbb{C}^{2^3}$, it is orthogonal to seven linearly independent vectors corresponding to seven different rows. By permutating rows, we may assume that $\ket{\phi_5}$ is orthogonal to the first seven rows. Then there exists a product vector $\ket{g_{8,1}',g_{9,2}',g_{10,3}',g_{11,4}',\phi_5}$ with $\ket{\phi_1,\phi_2,\phi_3,\phi_4}=\ket{g_{8,1}',g_{9,2}',g_{10,3}',g_{11,4}'}$ such that it is orthogonal to all rows of $A$. So we obtain that for the case $(1,1,1,1,3)$, the $7$-qubit UPB of size $11$ is not a UPB in this five partitions.
	
	For the case $(1,1,1,2,2)$, it implies that the product vectors of $A$  in \eqref{uom:1} are in space $\mathbb{C}^2\otimes\mathbb{C}^2\otimes\mathbb{C}^2\otimes\mathbb{C}^4\otimes\mathbb{C}^{4}$. It is known that if a set of product vector is not a UPB in space  $\mathbb{C}^2\otimes\mathbb{C}^2\otimes\mathbb{C}^2\otimes\mathbb{C}^2\otimes\mathbb{C}^2\otimes\mathbb{C}^4$, then it is also not a UPB in space $\mathbb{C}^2\otimes\mathbb{C}^2\otimes\mathbb{C}^2\otimes\mathbb{C}^4\otimes\mathbb{C}^4$. In Theorem \ref{thm:main}, it has shown that $7$-qubit UPB is still a UPB in six partitions by merging systems  $A_2,A_3$ and $A_1,A_2$. We further consider whether it is a UPB in this five partitions. 
	Up to systems permutation of $A$, we may assume that the form of product vectors of $A$ is  
	\begin{eqnarray}
		\label{eq:mn23}
		\ket{g_{i,1},g_{i,2},g_{i,3},g_{i,4},g_{i,5}}&\in&\mathcal{H}_{A_l}\otimes\mathcal{H}_{A_j}\otimes\mathcal{H}_{A_k}\otimes\mathcal{H}_{A_mA_n}\otimes\mathcal{H}_{A_2A_3}\nonumber\\&=&\mathbb{C}^2\otimes\mathbb{C}^2\otimes\mathbb{C}^2\otimes\mathbb{C}^4\otimes\mathbb{C}^{4},
	\end{eqnarray}
	where $i=1,\cdots,11$ and $l,j,k,m,n\in\{1,\cdots,7\}\backslash\{2,3\}$ and they are distinct. Suppose that there exists a product vector $\ket{\phi_1,\cdots,\phi_5}$ such that it is orthogonal to all rows of $A$ up to permutations of the systems. According to Lemma \ref{le:3} (b), the result shows that $\ket{\phi_5}$ can be orthogonal to five rows of $A$ in \eqref{uom:1}, i.e., the two cases of five rows are $(1,3,6,7,10)$ and $(1,3,6,8,9)$. Then it implies that $\ket{\phi_1,\phi_2,\phi_3,\phi_4}\in\mathbb{C}^2\otimes\mathbb{C}^2\otimes\mathbb{C}^2\otimes\mathbb{C}^4$ is orthogonal to the remaining six rows. If $\ket{\phi_5}$ is orthogonal to row $1,3,6,7,10$, then we permutate the rows of $A$ such that the five rows are the first five rows.  Then we find a product vector $\ket{\phi_1,\cdots,\phi_5}$, which is orthogonal to all rows of $A$, where $\ket{\phi_1,\phi_2,\phi_3} = \ket{g_{6,1}',g_{7,2}',g_{8,3}'}$ and $\ket{\phi_4}\in \mathbb{C}^4$ is orthogonal to $\ket{g_{9,4}},\ket{g_{10,4}},\ket{g_{11,4}}$. So this case does not hold. Similarly, by using the same method, if $\ket{\phi_5}$ is orthogonal to row $1,3,6,8,9$, one can verify that there exists a product vector such that it is orthogonal to all rows of $A$. So this case does not hold.
	
	
	Then we consider the case that one of the merged systems is $A_1,A_2$. Above Eq. \eqref{eq:mn23}, we know that if the set of product vectors in space  $\mathcal{H}_{A_l}\otimes\cdots\otimes\mathcal{H}_{A_1}\otimes\mathcal{H}_{A_2}\otimes\mathcal{H}_{A_mA_n}$ is not a UPB, then it is also not a UPB in space  $\mathcal{H}_{A_l}\otimes\cdots\otimes\mathcal{H}_{A_1A_2}\otimes\mathcal{H}_{A_mA_n}$, where $l,m,n\in\{1,\cdots,7\}\backslash\{1,2\}$. In Theorem \ref{thm:main}, we have shown that $7$-qubit UPB is not a UPB in space $\mathcal{H}_{A_l}\otimes\cdots\otimes\mathcal{H}_{A_1}\otimes\mathcal{H}_{A_2}\otimes\mathcal{H}_{A_mA_n}$. So we obtain that in this case, the $7$-qubit UPB of size $11$ is not a five-partite UPB.
	
	Hence,  the $7$-qubit UPB of size $11$ is not a UPB in any five partitions.
\end{proof}

This lemma shows that there are no five-partite UPBs. In the following Lemma \ref{le:6} - \ref{le:two}, we will consider the existence of four, three-partite and bipartite UPBs constructed by the $7$-qubit UPB of size $11$.

\begin{lemma}
	\label{le:6}
	The $7$-qubit UPB of size $11$ is not a UPB in any four partitions.
\end{lemma}
\begin{proof}
	We denote $(s,x,y,z)$ as the number of merging systems in four partitions. Since the $7$-qubit UPB of size $11$ is on seven systems, up to systems permutation, we obtain three cases $(1,1,1,4),(1,1,2,3)$ and $(1,2,2,2)$. For the first case $(1,1,1,4)$, one can verify that there exists a product vector $\ket{\phi_1,\phi_2,\phi_3,\phi_4}\in\mathbb{C}^2\otimes\mathbb{C}^2\otimes\mathbb{C}^2\otimes\mathbb{C}^{2^4}$ such that it is orthogonal to all rows of $A$. Because the number of the orthonormal vectors in space $\mathbb{C}^4$ is $2^4>11$, we obtain that there exists a $\ket{\phi_4}\in\mathbb{C}^{2^4}$, which is orthogonal to all rows of $A$. So in this case, the $7$-qubit UPB of size $11$ is not a UPB. For case $(1,1,2,3)$, up to systems permutation of $A$ in \eqref{uom:1}, it implies that the product vectors of $A$ are in space $\mathbb{C}^2\otimes\mathbb{C}^2\otimes\mathbb{C}^4\otimes\mathbb{C}^8$. 
	
	If a set of product vectors in space  $\mathbb{C}^2\otimes\mathbb{C}^2\otimes\mathbb{C}^2\otimes\mathbb{C}^2\otimes\mathbb{C}^8$ is not a UPB, then it is not  a UPB in space $\mathbb{C}^2\otimes\mathbb{C}^2\otimes\mathbb{C}^4\otimes\mathbb{C}^8$. From the results in Lemma \ref{le:five} below \eqref{eq:11113},  the $7$-qubit UPB of size $11$ is not a UPB in this case. Similarly, for case $(1,2,2,2)$, 
	according to the results in Lemma \ref{le:five}, it has shown that the $7$-qubit UPB of size $11$ is not a UPB in space $\mathbb{C}^2\otimes\mathbb{C}^2\otimes\mathbb{C}^2\otimes\mathbb{C}^4\otimes\mathbb{C}^{4}$, then it is also not a UPB in space $\mathbb{C}^2\otimes\mathbb{C}^4\otimes\mathbb{C}^4\otimes\mathbb{C}^4$. Thus,  the $7$-qubit UPB of size $11$ is not a UPB in any four partitions.
\end{proof}

\begin{lemma}
	\label{le:7}
	The $7$-qubit UPB of size $11$ is not a UPB in any three partitions.
\end{lemma}

\begin{proof}
	We denote $(x,y,z)$ as the number of merging systems in three partitions.  Since the $7$-qubit UPB of size $11$ is on seven systems, up to systems permutation, we obtain four cases $(1,1,5),(1,2,4),(1,3,3)$ and $(2,2,3)$.
	
	For the case $(1,1,5)$, it implies that the product vectors of $A$ in \eqref{uom:1} are in space $\mathbb{C}^2\otimes\mathbb{C}^2\otimes\mathbb{C}^{2^{5}}$. Furthermore, if we merge any five systems of $A_i,i=1,\cdots,7$  then evidently they are in $\mathbb{C}^{2^{5}}$. For example, when we merge systems $A_3,\cdots,A_7$, the set
	of $7$-qubit product vectors corresponding to $A$ in \eqref{uom:1} become product vectors in space  $\mathbb{C}^2\otimes\mathbb{C}^2\otimes\mathbb{C}^{2^{5}}=\mathcal{H}_{A_1}\otimes\mathcal{H}_{A_2}\otimes\mathcal{H}_{A_3,\cdots,A_7}$. We may assume that the product vectors of $A$ are  $\ket{g_j,h_j,f_j}\in\mathbb{C}^2\otimes\mathbb{C}^2\otimes\mathbb{C}^{2^{5}},j=1,\cdots,11$. Because the number of the  orthonormal vectors in space $\mathbb{C}^{2^{5}}$ is $32>11$, there must exist a vector $\ket{x_1}$ such that it is orthogonal to every $\ket{f_j}$ and $\ket{x_1}\neq\ket{f_j},j=1,\cdots,11$. So regardless of  the five merged systems in $A$, the case $(1,1,5)$ does not hold. Similarly, in case $(1,2,4)$, we assume that the product vectors of $A$ are  $\ket{g_j,h_j,f_j}\in\mathbb{C}^2\otimes\mathbb{C}^{2^{2}}\otimes\mathbb{C}^{2^{4}},j=1,\cdots,11$. Because the number of the  orthonormal vectors in space $\mathbb{C}^{2^{4}}$ is $16>11$, there must exist a vector $\ket{x_1}\in\mathbb{C}^{2^{4}}$ such that it is orthogonal to every $\ket{f_j}$ and $\ket{x_1}\neq\ket{f_j},j=1,\cdots,11$. So we obtain that the case $(1,2,4)$ does not hold.
	
	For the case $(1,3,3)$, we assume that the product vectors of $A$ in \eqref{uom:1}  are  $\ket{g_j,h_j,f_j}\in\mathbb{C}^2\otimes\mathbb{C}^{2^{3}}\otimes\mathbb{C}^{2^{3}}$ with $\ket{h_j}=\ket{a_j,b_j,c_j},j=1,\cdots,11$. Then for any four product vectors $\ket{g_{m_1},h_{m_1},f_{m_1}},\ket{g_{m_2},h_{m_2},f_{m_2}},\ket{g_{m_3},h_{m_3},f_{m_3}}$ and $\ket{g_{m_4},h_{m_4},f_{m_4}}$, there exists a product vector $\ket{g_{m_1}',a_{m_2}',b_{m_3}',c_{m_4}',x_1}$ such that it is orthogonal to the four product vectors, where $m_1,m_2,m_3,m_4\in\{1,\cdots,11\}$ and they are distinct. 
	We consider whether $\ket{x_1}$ is orthogonal to the remaining seven $\ket{f_j}$, $j\in\{1,\cdots,11\}\backslash\{m_1,m_2,m_3,m_4\}$. Since $\ket{f_j}\in\mathbb{C}^{2^{3}}$, in space $\mathbb{C}^{2^3}$, there are eight orthonormal vectors. Then there exists a vector $\ket{x_1}$ such that $\ket{x_1}$ is orthogonal to the remaining seven $\ket{f_j}$. So  $\ket{g_{m_1}',a_{m_2}',b_{m_3}',c_{m_4}',x_1}$ is orthogonal to all rows of $A$. The case $(1,3,3)$ does not hold.
	
	For the case $(2,2,3)$, we assume that the product vectors of $A$ in \eqref{uom:1}  are  $\ket{g_j,h_j,f_j}\in\mathbb{C}^{2^{2}}\otimes\mathbb{C}^{2^{2}}\otimes\mathbb{C}^{2^{3}}$ with $\ket{g_j}=\ket{a_j,b_j}$ and  $\ket{h_j}=\ket{c_j,d_j}, j=1,\cdots,11$.  Then for any four product vectors $\ket{g_{m_1},h_{m_1},f_{m_1}},\ket{g_{m_2},h_{m_2},f_{m_2}},\ket{g_{m_3},h_{m_3},f_{m_3}}$ and $\ket{g_{m_4},h_{m_4},f_{m_4}}$, there exists a product vector $\ket{a_{m_1}',b_{m_2}',c_{m_3}',d_{m_4}',x_1}$ such that it is orthogonal to the four product vectors, where $m_1,m_2,m_3,m_4\in\{1,\cdots,11\}$ and they are distinct. 
	We consider whether $\ket{x_1}$ is orthogonal to the remaining seven $\ket{f_j}$, $j\in\{1,\cdots,11\}\backslash\{m_1,m_2,m_3,m_4\}$. Since $\ket{f_j}\in\mathbb{C}^{2^{3}}$, in space $\mathbb{C}^{2^3}$, there are eight orthonormal vectors. Then there exists a vector $\ket{x_1}$ such that $\ket{x_1}$ is orthogonal to the remaining $\ket{f_j}$. So  $\ket{a_{m_1}',b_{m_2}',c_{m_3}',d_{m_4}',x_1}$ is orthogonal to the product vectors of $A$. The case $(2,2,3)$ does not hold.
	
	According to above cases, the $7$-qubit UPB of size $11$ is not a UPB in any three partitions.
\end{proof}

\begin{lemma}
	\label{le:two}
	Then $7$-qubit UPB is not a UPB in any bipartitions.
\end{lemma}
\begin{proof}
	For bipartitions of the $7$-qubit UPB of size $11$ in \eqref{uom:1}, up to systems permutation, we obtain that the number of merged systems in bipartition is $(1,6),(2,5)$ or $(3,4)$. Suppose that the product vectors of $A$ in bipartition are $\ket{a_i,b_i},i=1,\cdots,11$. We consider whether there exists a product vector $\ket{\phi_1,\phi_2}$ such that it is orthogonal to $\ket{a_i,b_i}$. For case $(3,4)$, we have $\ket{\phi_1,\phi_2}\in \mathbb{C}^{2^3}\otimes\mathbb{C}^{2^4}$. Since the dimension of space $\mathbb{C}^{2^4}$ is $16>11$, there exists a $\ket{\phi_2}\in\mathbb{C}^{2^4}$ such that it is orthogonal to all $\ket{b_i}$. So we find a product vector $\ket{\phi_1,\phi_2}$.
	
	If one of the dimensions of $\ket{\phi_1},\ket{\phi_2}$ is larger than $11$, then there exists a $\ket{\phi_1,\phi_2}$ such that it is orthogonal to all $\ket{a_i,b_i}$. One can verify that the three cases $(1,6),(2,5),(3,4)$ all hold. So the $7$-qubit UPB of size $11$ is not a UPB in any bipartitions. 
\end{proof}

\bibliographystyle{unsrt}

\bibliography{20210826con}

\end{document}